\documentclass[a4paper,UKenglish,numberwithinsect,cleveref, autoref, thm-restate]{lipics-v2021}
\usepackage{xspace,color,amsmath,amsthm,amssymb,graphicx,tikz-cd}

\title{Right-Adjoints for Datalog Programs, and Homomorphism Dualities over Restricted Classes}
\titlerunning{Right-Adjoints for Datalog Programs, and Homomorphism Dualities}

\authorrunning{Balder ten Cate, V\'{i}ctor Dalmau, and Jakub Opr\v{s}al}
\Copyright{Balder ten Cate, V\'{i}ctor Dalmau, and Jakub Opr\v{s}al}

\author{Balder ten Cate}{Institute for Logic, Language, and Computation, University of Amsterdam, The Netherlands}{b.d.tencate@uva.nl}{https://orcid.org/0000-0002-2538-5846}{Supported by the European Union’s Horizon 2020 research and innovation programme (MSCA-101031081).}

\author{V\'{i}ctor Dalmau}{Department of Information and Communication Technologies, Universitat Pompeu Fabra, Spain}{victor.dalmau@upf.edu}{https://orcid.org/0000-0002-9365-7372}{}

\author{Jakub Opr\v{s}al}
       {Institute of Science and Technology Austria, Klosterneuburg, Austria}
       {jakub.oprsal@ist.ac.at}
       {https://orcid.org/0000-0003-1245-3456}
       {Supported by the European Union’s Horizon 2020 research and innovation programme under the Marie Skłodowska-Curie Grant Agreement No 101034413.}

\nolinenumbers

\ccsdesc[500]{Theory of computation~Logic and databases}

\keywords{Datalog, Adjoints, Homomorphism Dualities, Database Constraints, Conjunctive Queries}

\theoremstyle{definition}
\newtheorem{contribution}{Contribution}
\theoremstyle{claimstyle}
\newtheorem{myclaim}{Claim}
\newtheorem{subclaim}{Subclaim}

\newcommand{\chase}{\textrm{chase}}
\newcommand{\edatalog}{\text{$\exists$Datalog}\xspace}
\newcommand{\bfS}{\textbf{S}}
\newcommand{\bfX}{\textbf{X}}
\newcommand{\colondash}{\mathrel{{:}{-}}} %
\newcommand{\bx}{\mathbf{x}}
\newcommand{\by}{\mathbf{y}}

\newcommand{\Inst}{\text{\rm Inst}}
\newcommand{\Instinf}{\text{\rm Inst$^\infty$}}

\newcommand{\Unfoldings}{\operatorname{Unfoldings}}

\newcommand{\upclosure}{\mathop{\!\uparrow}}
\newcommand{\downclosure}{\mathop{\!\downarrow}}

\let\paragr\subparagraph

\begin{document}

\maketitle

\begin{abstract}
A Datalog program can be viewed as a syntactic specification of a functor from database instances over some schema to database instances over another schema.
The same holds more generally for 
\edatalog. We establish large classes
of Datalog and $\edatalog$ programs for which the corresponding functor admits a generalized right-adjoint. We employ these results to obtain new insights into the existence of and methods for constructing homomorphism dualities within restricted classes of instances. We also derive new results regarding the existence of uniquely characterizing data examples for database queries.
\end{abstract}

\section{Introduction}

Datalog is a rule-based language for specifying mappings from database
instances over an input schema $\bfS_{in}$, to database instances over an output schema
$\bfS_{out}$.

\begin{example}\label{ex:tc}
Consider the Datalog program defined by the 
following rules:
\begin{align*}
Path(x,y) &\colondash Edge(x,y). \\
Path(x,y) &\colondash Edge(x,z), Path(z,y). \\
Ans(x,y) &\colondash Path(x,y).
\end{align*}
This Datalog program takes as input an  instance over an input schema $\{Edge\}$, and produces as output an  instance 
over the schema $\{Ans\}$, where $Ans$ is the transitive
closure of $Edge$. 
\end{example}

Using terminology from category theory, a Datalog program defines a functor from $\Inst[\bfS_{in}]$ to $\Inst[\bfS_{out}]$, where $\Inst[\bfS]$ denotes the category of all database instances over schema $\bfS$, with homomorphisms as the arrows.

We study the existence of right-adjoints and generalized right-adjoints for such a functor. We define a \emph{right-adjoint} for an arbitrary functor $F:X\to Y$ as a functor $G:Y\to X$ such that, for all $A\in X$ and $B\in Y$,  $F(A)\to B$ iff $A\to G(B)$. 
Loosely speaking, \emph{generalized right-adjoints} are defined similarly, except that we allow $G$ to map an object $B\in Y$ to a \emph{finite set} of objects in $X$, 
 such that, for all $A\in X$ and $B\in Y$,  $F(A)\to B$ iff $A\to B'$ for some $B'\in G(B)$.
As it turns out, the Datalog program $P$ from Example~\ref{ex:tc} has a right-adjoint. There are also Datalog programs that do not have a right-adjoint but that have a generalized right-adjoint, and Datalog programs that do not have a generalized right-adjoint.

\begin{contribution}[Section~\ref{sec:tam}]
We introduce a new fragment of 
Datalog called \emph{TAM Datalog} (which stands for \emph{Tree-Shaped Almost-Monadic Datalog}). We  characterize TAM Datalog semantically
as a fragment of Monadic Second-Order Logic, and we prove that TAM Datalog
is closed under composition.
\end{contribution}

\begin{contribution}[Section~\ref{sec:adjoints}]
We show that every connected TAM Datalog program has a right-adjoint, 
and that every TAM Datalog program has a generalized right-adjoint.
We show by means of counterexamples that each of the syntactic conditions
imposed by TAM Datalog is necessary for the existence of generalized right-adjoints.
We also
identify a larger fragment of \edatalog (the extension of Datalog with existential quantifiers) that admits generalized right-adjoints
\end{contribution}

Our motivation for studying (generalized) right-adjoints comes from the
fact that they provide us with a means of constructing homomorphism dualities.
A homomorphism duality is a pair $(F,D)$ where $F$ and $D$ are sets of
instances, such that an arbitrary instance $A$ admits a homomorphism
from a instance in $F$ if and only if $A$ does not admit a homomorphism
to any instance in $D$. In other words, homomorphism dualities 
equate the existence of a homomorphism of one kind to the non-existence
of a homomorphism of another kind. Homomorphism dualities have been 
studied extensively in the literature on constraint satisfaction problems,
and have also found applications elsewhere, e.g., in database theory
and knowledge representation.
\looseness=-1

\begin{contribution}[Section~\ref{sec:dualities}]
We show that generalized right-adjoints can be used to construct homomorphism dualities, and we obtain new results regarding the existence of finite homomorphism dualities within restricted classes of instances, e.g., transitive digraphs.
\end{contribution}

\begin{contribution}[Section~\ref{sec:applications}]
In~\cite{AlexeCKT2011,tCD2022:conjunctive}, 
homomorphism dualities are used as a tool
for studying the unique characterizability,
and exact learnability, of conjunctive queries 
and unions of conjunctive queries. Following
this approach, we derive new results on
the unique characterizability of unions of 
conjunctive queries in the presence of a 
background theory, addressing an open question
from~\cite{tCD2022:conjunctive}.
\end{contribution}

\paragr{Related Work}
Foniok and Tardif \cite{Foniok2015:functors} studied existence of right adjoints to Pultr functors which are themselves right adjoints  \cite{Pultr70} in the special case of digraphs. Translating into our terms a Pultr functor is a interpretation (of digraphs in digraphs) $(\phi_V,\phi_E)$ where $\phi_V$ and $\phi_E$ are conjunctive queries (with $k$ and $2k$ free variables, respectively, for some $k\geq 1$) defining the output node-set and edge-set respectively. 
For the special case where $\phi_V$ just returns the input node-set, it was shown in \cite{Foniok2015:functors} that the functor defined by $(\phi_V, \phi_E)$ has a right adjoint if and only if $\phi_E$ is connected and acyclic.
The setup and characterization were generalized in \cite{DalmauKO} to  arbitrary relational structures. We extensively build on the framework in and concepts in~\cite{DalmauKO}, but we permit the interpretation to be specified by an arbitrary $\edatalog$ program, so that our setup is able to encompass common types of database dependencies such as inclusion dependencies.
Our set-up based on \edatalog can be viewed as a generalization of the one in  \cite{DalmauKO} (cf.~Appendix~\ref{app:pultr}).

To the best of our knowledge, this is the first time that functors defined by Datalog programs have been studied. Also, it is the first application of functors with right adjoints in the context of unique characterization of UCQs. In a different setting, namely approximate graph coloring, the so-called arc graph functor was used in \cite{KrokhinOWZ20} where it is additionally argued that functors with right adjoint can, more generally, play a role in the design and analysis of reductions between promise constraint satisfaction problems. The use of Datalog programs for reductions between such problems is discussed in \cite{DalmauO23}.%

\section{Preliminaries}
\label{sec:prel}

\paragr{Schemas, Instances, Homomorphisms}
A \emph{schema} $\bfS$ is a finite collection of relation symbols $R$ with
specified arity $arity(R)\geq 0$. An $\bfS$-\emph{instance} $I$ is a 
 set of facts, where a fact is an expression of the form
$R(a_1, \ldots, a_n)$ with $R\in\bfS$ and $n=arity(R)$. 
Unless specified otherwise, instances are always assumed to be finite.
The  \emph{active domain} $adom(I)$ of $I$ is the set of all values $a_i$ occurring in the facts of $I$.
A \emph{homomorphism} $h:I\to J$, where $I$ and $J$ are instances over the 
same schema $\bfS$, is a function from $adom(I)$ to $adom(J)$ such that
the $h$-image of every fact of $I$ is a fact of $J$.

We will denote by 
$\Inst[\bfS]$  the set of all finite $\bfS$-instances, and define $\Instinf[\bfS]$ similarly, except allowing also infinite
$\bfS$-instances.
Category theoretically, we can view $\Inst[\bfS]$ and $\Instinf[\bfS]$ also as categories. In 
this case, the objects are the instances and the arrows are homomorphisms. 
This will allow us to speak, for example, of
\emph{functors} from $\Inst[\bfS]$ to $\Inst[\bfS']$.

A $k$-ary \emph{pointed $\bfS$-instance} (for $k\geq 0$) is a 
pair $(I,\textbf{a})$ where $I$ is an $\bfS$-instance and 
$\textbf{a}$ a $k$-tuple of elements of $adom(I)$,
called \emph{distinguished elements}. A
homomorphism $h:(I,\textbf{a})\to (J,\textbf{b})$ is
a homomorphism $h:I\to J$ such that $h(\textbf{a})=\textbf{b}$.

\paragr{Incidence Graph, Connectedness, C-Acyclicity}
The \emph{incidence graph} of an instance $I$
is the bipartite multi-graph whose nodes are the elements and 
the facts of $I$, and where there is a distinct (undirected) edge
$(a,f)$ for every occurrence of the element $a$ in the fact.
We say that an instance is \emph{connected} if its incidence graph
is connected, and an instance is \emph{acyclic} if its
its incidence graph is acyclic.
A pointed instance $(I,\textbf{a})$ is \emph{c-acyclic} if 
every cycle in the incidence graph of $I$ contains
at least one element from the tuple $\textbf{a}$.

\paragr{Datalog}
A Datalog program is specified by a collection of rules,
and it defines a mapping from instances over a schema
$\bfS_{in}$ (traditionally known as the EDB schema)
to instances over a schema $\bfS_{out}$ (traditionally
known as the IDB schema). The presentation we will give here also
allows for auxiliary IDB relations that are not exposed in the
output schema.

\begin{definition}[Datalog program]
A \emph{Datalog program} is a tuple 
$P=(\bfS_{in},\bfS_{out},\bfS_{aux},\Sigma)$ where
$\bfS_{in},\bfS_{out},\bfS_{aux}$ are mutually disjoint schemas, 
and $\Sigma$ is a set of rules of the form
\[ R_0(\textbf{x}_0) \colondash R_1(\textbf{x}_1), \ldots, R_n(\textbf{x}_n) \]
where 
$R_0\in \bfS_{out}\cup\bfS_{aux}$, 
$R_1, \ldots, R_n\in \bfS_{in}\cup\bfS_{aux}$, 
and $\{\textbf{x}_0\}\subseteq \{\textbf{x}_1, \ldots, \textbf{x}_n\}$.
 \end{definition}

If $P$ is a Datalog program, then we will use often
use the notation $\bfS_{in}^P$, $\bfS_{out}^P$, $\bfS_{aux}^P$, and $\Sigma^P$ to refer to the 
constituents of the tuple $P$. 

The \emph{head} of a rule is the part 
to the left of the $\colondash$ sign, and
the \emph{body} is the part to the right.
The \emph{canonical instance} of a Datalog rule
$R_0(\textbf{x}_0) \colondash R_1(\textbf{x}_1), \ldots, R_n(\textbf{x}_n)$
is the pointed instance whose active domain is $\{\textbf{x}_1, \ldots, \textbf{x}_n\}$,
whose facts are the conjuncts of the rule body, and whose sequence of distinguished
elements is the tuple $\textbf{x}_0$.
 We say that a Datalog program $P$ is  \emph{connected} if the 
 canonical instance of each rule is connected.

If $P$ is a Datalog program and $I$ a $\bfS_{in}^P$-instance, then
a \emph{solution} for $I$ with respect to $P$ is an
instance $J$ over the schema $\bfS_{in}\cup\bfS_{out}\cup\bfS_{aux}$ 
such that $I\subseteq J$, and such that all the rules of $P$ are 
satisfied in $J$ (i.e., whenever the body of a rule is satisfied, then so is the head). The well-known 
\emph{chase} procedure provides a method for constructing
a solution: given a Datalog program $P$
and an $\bfS_{in}^P$-instance $I$, we denote by $\chase_P(I)$ the 
$\bfS_{in}^P\cup\bfS_{out}^P\cup\bfS_{aux}^P$-instance
obtained from $I$ by applying all rules until convergence.
More precisely, $\chase_P(I)$ can be defined as the infinite
union $\bigcup_{i\geq 0} \chase^i_P(I)$, where $chase^0_P(I)=I$, and 
where $\chase^{i+1}_P(I)$ extends $\chase^i_P(I)$ with all facts
that can be derived from facts in $\chase^i_P(I)$ using a
rule in $\Sigma^P$. We refer to \cite{Alice} for more details.

\begin{lemma}
For all Datalog programs $P$ and $\bfS_{in}^P$-instances $I$, $\chase_P(I)$ is a solution for $I$ with respect to $P$. Moreover, it is the intersection of all solutions
for $I$ with respect to $P$.
\end{lemma}

We denote the
 $\bfS_{out}^P$-reduct of $\chase_P(I)$ by $P(I)$.

By a \emph{Boolean} Datalog program, we mean a Datalog program $P$
where $\bfS_{out}^P$ consists of a single zero-ary relation symbol,
which is customarily denoted as \emph{Ans}. 
In such cases, write $P(I)=true$ if
$P(I)=\{Ans()\}$ and $P(I)=false$ otherwise (i.e., if $P(I)=\emptyset$).

We can think of the above definition of $P(I)$, in terms of the chase, as a bottom-up account of the 
semantics of a Datalog program. \emph{Unfoldings} (a.k.a.~\emph{expansions}) provide a complementary, top-down account. 
Given a Datalog program $P$, the set of \emph{derivable rules}
of $P$ is the smallest set of rules that (i) contains all
rules of $P$, and (ii) is closed under the operation of 
substituting occurrences of rule heads by the corresponding rule bodies (renaming variables as necessary). 
Given a Datalog program $P$ and a relation $R\in\bfS_{out}^P$, 
$\Unfoldings(P,R)$ is the set of canonical instances of derivable rules that have $R$ in the rule head and that only have
 $\bfS_{in}$-relations in the body.
Note that this set is in general infinite.

\begin{example}
Let $P$ be the Datalog program consisting of the three rules
\[R(x,y) \colondash S(x,y) \qquad
R(x,x) \colondash T(x,y) \qquad
T(x,y) \colondash U(x,y), U(y,z)
\]
where $\bfS_{in}=\{U, S\}$, $\bfS_{out}=\{R\}$, and $\bfS_{aux} = \{T\}$.
Then $\Unfoldings(P,R)$ consists (up to isomorphism) of the pointed instances
$(\{U(a,b),U(b,c)\}, \langle a,a\rangle)$ and $(\{S(a,b)\}, \langle a,b\rangle)$.
\end{example}

\begin{lemma}[Cf.~\cite{CV97:equivalence}]\label{lem:unfoldings}
For all Datalog programs $P$, instances $I\in \Inst[\bfS_{in}^P]$, and
 $\bfS_{out}^P$-facts $R(\textbf{a})$ over $adom(I)$,
 $R(\textbf{a})\in P(I)$ iff,
for some $(J,\textbf{b})\in \Unfoldings(P,R)$, $(J,\textbf{b})\to (I,\textbf{a})$.
\end{lemma}

\paragr{\edatalog}
The language of \edatalog extends Datalog with 
existential quantifiers.

\begin{definition}[\edatalog]
An \emph{\edatalog rule} is an expression of the 
form 
\[ \exists \textbf{z} \big(R_1(\textbf{x}_1), \ldots, R_n(\textbf{x}_n)\big) \colondash
   S_1(\textbf{y}_1), \ldots,
   S_m(\textbf{y}_m) 
\]
where $\{\textbf{x}_i\}\subseteq \{\textbf{y}_1, \ldots, \textbf{y}_n, \textbf{z}\}$.
In the context of such a rule, 
the variables in $\textbf{z}$ are called
\emph{existential variables}. 
If the tuple $\textbf{z}$ is non-empty, 
we will also 
call the rule an \emph{existential rule}.
An \emph{exported variable} is a variable occurring 
both in the body and in the head of the rule.

An \emph{\edatalog program} is a tuple
$P=(\bfS_{in}, \bfS_{out}, \bfS_{aux}, \Sigma)$, where
$\bfS_{in}, \bfS_{out}, \bfS_{aux}$ are disjoint 
schemas and $\Sigma$ is a set of
\edatalog rules, where each
relation occurring in the body of a rule is from $\bfS_{in}\cup\bfS_{aux}$, and each relation occurring in the 
head of a rule is from $\bfS_{out}\cup\bfS_{aux}$.
\end{definition}

Just as in the case of Datalog, a \emph{solution}
for a $\bfS_{in}$-instance $I$ with respect to an
\edatalog-program $P$ is an instance $J$ over the 
schema $\bfS_{in}\cup\bfS_{out}\cup\bfS_{aux}$
such that $I\subseteq J$ and such that all the rules
of $P$ are satisfied in $J$. 
However, unlike in the case of Datalog, we now allow
for solutions to be infinite, for reasons that will
become clear in a moment.
A \emph{universal solution}
for $I$ (w.r.t.~$P$) is a solution $J$ for $I$ such that
for every solution $J'$ for $I$, it holds that
$J\to_{adom(I)} J'$.
Here, as a convenient notation, we  write $h:I\to_X J$ if $h:I\to J$ 
and $h(x)=x$ for all $x\in X$. We will also write $I\leftrightarrow_X I'$
if $I\to_X I'$ and $I'\to_X I$.
It is well known that every
instance $I$ has a (possibly infinite) universal solution,
and that universal solutions are unique up to homomorphic
equivalence. More precisely, if $J$ and $J'$ are universal
solutions for the same $\bfS_{in}$-instance $I$, then
$J\leftrightarrow_{adom(I)} J'$~(cf.~\cite{Cali2012}).

We will use the notation
$P(I)$ to denote the $\bfS_{out}$-reduct of an arbitrary
universal solution of $I$. 
This uniquely defines $P(I)$ up to 
$\leftrightarrow_{adom(I)}$-equivalence.

\begin{example}
Let $P$ be the \edatalog program consisting of the three rules
\[ R(x,y)\colondash R_{in}(x,y) ~~~~~~~~
  \exists z R(y,z) \colondash R(x,y) ~~~~~~~~
  R_{out}(x,y) \colondash R(x,y).
\]
The instance $I=\{R_{in}(a_1,a_2)\}$ does not have a finite universal solution with respect to $P$, but has an infinite
universal solution, namely
$J=I\cup\{R(a_i, a_{i+1}), R_{out}(a_i,a_{i+1})\mid i=1,2,\ldots\}$.
\end{example}

Since $P(I)$ is, in general, infinite, it is common to  
impose additional restrictions $P$ to ensure that $P(I)$
is finite. One well-known such restriction is 
\emph{weak acyclicity} \cite{FKMPicdt03}.
We will omit the precise definition here
(cf.~Appendix~\ref{app:more-prels}).
Weak acyclicity ensures that finite universal solutions
exist and can be computed in polynomial time using a suitable version of the chase.

\begin{proposition}[\cite{FKMPicdt03}]
Fix an \edatalog program $P$.
If $P$ is weakly acyclic, then every finite $\bfS_{in}$-instance has a finite universal solution,
which can be computed in polynomial time.
\end{proposition}

We say
that an \edatalog program $P$ is 
\emph{non-recursive} if $\bfS^P_{aux}=\emptyset$.
Every non-recursive \edatalog program is weakly
acyclic.

\begin{lemma}\label{lem:edatalog-monotonicity}
Let $P$ be any \edatalog
program, and let $I,I'$ be $\bfS_{in}^P$-instances. Every homomorphism
$h:I\to I'$ extends to a homomorphism $h':P(I)\to P(I')$. 
\end{lemma}

We say that two \edatalog programs $P, P'$ with $\bfS_{in}^P=\bfS_{in}^{P'}$ and
$\bfS_{out}^P=\bfS_{out}^{P'}$
are \emph{equivalent} if, for all 
$\bfS_{in}^P$-instances $I$, 
$P(I)\leftrightarrow_{adom(I)} P'(I)$.

\section{TAM Datalog}
\label{sec:tam}

TAM Datalog is a fragment of Datalog defined by two requirements:
``tree-shaped'' and ``almost monadic''. We introduce each in isolation first.

\paragr{Almost-Monadic Datalog Programs}
Recall that a Datalog program is \emph{monadic} if all relations
 in $\bfS_{aux}$ are unary. 
It is well known that monadic Datalog programs can be expressed
in Monadic Second-Order logic (MSO).
Formally, by a 
\emph{$k$-ary MSO query} over a schema $\bfS$,
we will mean an MSO formula
$\phi(x_1, \ldots, x_k)$ over $\bfS$. 
We say that a Datalog program 
$P=(\bfS_{in},\bfS_{out},\bfS_{aux},\Sigma)$ together with a designated $k$-ary relation
$R\in \bfS_{out}$, \emph{defines} 
an MSO query $\phi_R(x_1, \ldots, x_k)$ over $\bfS_{in}$, if for all $\bfS_{in}$-instances $I$ and $a_1, \ldots, a_k\in adom(I)$,
$R(a_1, \ldots, a_k)\in P(I)$ iff $I\models\phi_R(a_1, \ldots, a_k)$.
The following is folklore in the database literature (cf.~\cite{Gottlob2004:monadic} for an explicit proof):
\begin{theorem}
 \label{thm:monadic-mso}
Let $P$ be a monadic Datalog
program and $R\in \bfS_{out}^P$. Then $(P,R)$ defines an
MSO query.
\end{theorem}

We will now define a weaker restriction, namely that of \emph{almost monadic} Datalog programs,
for which the same holds. These are programs in which
every $k$-ary auxiliary relation has, among its $k$ argument positions, (at most) one specified ``\emph{articulation position}'',
 and the syntax of the rules is
constrained in such a way that variables occurring in non-articulation positions can only be used to carry information forward, and not 
to perform joins. 

\begin{definition}[Almost-Monadic Datalog Programs]
An \emph{articulation function}, for a Datalog program $P$, 
is a partial function $f$ mapping relations $R\in \bfS_{aux}^P$ to
a number $f(R)\in\{1, \ldots, arity(R)\}$, which we will
refer to as the \emph{articulation position} of $R$. 
Each $i\in \{1, \ldots, arity(R)\}$ other than $f(R)$ is
called a \emph{non-articulation position} of $R$.
A Datalog program is \emph{almost monadic} if there exists an 
articulation function such that, in every
rule, each variable occurring in a non-articulation position of 
an auxiliary relation in a rule body occurs only once in that rule body, and
does not occur in the articulation position of any auxiliary relation in the head.
\end{definition}

Note: the articulation conditions pertain to 
auxiliary relations and not to output relations.

\begin{example}\label{ex:almost-monadic}
The Datalog program from Example~\ref{ex:tc}
(which computes all pairs $(a,b)$ for which 
there is a directed path from $a$ to $b$) is an almost-monadic Datalog program: 
the witnessing articulation function assigns to 
the auxiliary relation $R$ its first position as articulation 
position.
It is worth pointing out that, if we extend
the program with an additional rule
$Ans(x,y) \colondash Path(y,x)$
(so that it computes all pairs $(a,b)$ for which 
there is a directed path from $a$ to $b$ or from $b$ to $a$),
the resulting program is still almost-monadic. This is because the
requirements on the articulation function only pertain to 
auxiliary relations in the head, and not to output relations.
\end{example}
\begin{restatable}{proposition}{propalmostmonadic}
The almost-monadic Datalog program from  Example~\ref{ex:almost-monadic}
is not equivalent to a monadic Datalog program.
\end{restatable}

The following result justifies the terminology \emph{almost monadic}. It shows
that almost-monadic Datalog programs can be
simulated, in a precise sense, by monadic Datalog
programs.

\begin{restatable}{theorem}{thmmonadicreduction} 
\label{thm:monadic-reduction}
For each almost-monadic Datalog program $P$
and
$k$-ary relation symbol $R\in \bfS_{out}^P$, there is a Boolean
monadic Datalog program $P'$
where $\bfS_{in}^{P'}= \bfS_{in}^P\cup\{Q_1, \ldots, Q_k\}$, such that the following are equivalent,
for all $\bfS_{in}^P$-instances $I$ and $a_1, \ldots, a_k\in adom(I)$:
\begin{enumerate}
    \item      $R(a_1, \ldots, a_k)\in P(I)$, 
    \item           $P'(I\cup\{Q_1(a_1), \ldots, Q_k(a_k)\})=true$.
\end{enumerate}
Conversely, for every Boolean monadic Datalog program $P'$ with $\bfS_{in}^{P'}= \bfS\cup\{Q_1, \ldots, Q_k\}$, where  each $Q_i$ is unary, there is a TAM Datalog program $P$ with $\bfS_{in}^{P}=\bfS$ and
$\bfS_{out}^{P} = \{R\}$,  such that the above equivalence holds.
\end{restatable}

It follows that almost-monadic Datalog is contained in MSO. That is,
we have the following analogue of Theorem~\ref{thm:monadic-mso} for
almost-monadic Datalog programs (cf.~Fig.~1):

\begin{restatable}{corollary}{coralmostmonadictomso}
\label{cor:almost-monadic-to-mso}
Let $P$ be an almost-monadic Datalog program and $R\in \bfS_{out}^P$. Then $(P,R)$ defines an MSO query.
\end{restatable}

\begin{figure}
\begin{center}
\usetikzlibrary{positioning,
                shapes.geometric}
\begin{tikzpicture}
 
\node[draw,
    ellipse,
    minimum height=3cm, 
    minimum width=6cm,
    label={165:MSO}] (circle1) at (0,0){};
 
\node[draw,
    ellipse,
    minimum height=3cm, 
    minimum width=6cm,
    label={15:Datalog}] (circle1) at (2,0){};

\node[draw,
    ellipse,
    fill=red!10,
    minimum height =2.5cm,
    minimum width = 3.5cm] (circle2) at (1,0){};
    
\node[draw,
    circle,
    fill=blue!20,
    minimum size =1.5cm] (circle2) at (1,0.4){};

 \node[label={[align=center]\small \begin{tabular}{cc}Monadic\\Datalog\end{tabular}}] at (1,-0.3) {};
 
 \node[label={[align=center]\small \begin{tabular}{cc}Almost Monadic\\Datalog\end{tabular}}] at (1,-1.45) {};

\end{tikzpicture}
\end{center}
\vspace{-4mm}
\caption{Almost Monadic Datalog in relation to MSO and Datalog.}
\vspace{-3mm}
\end{figure}
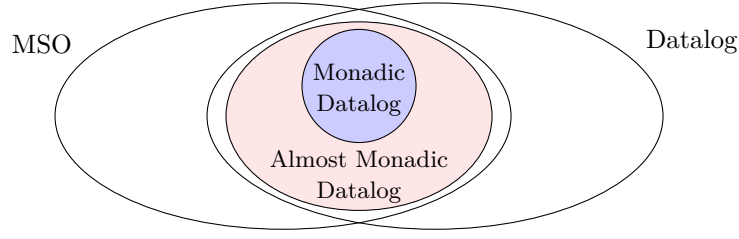

We do not know whether almost-monadic Datalog is strictly contained
in the intersection of MSO and Datalog. However, as we will soon see,
the intersection of \emph{tree-shaped} Datalog and MSO 
is (up to logical equivalence) precisely tree-shaped almost-monadic Datalog.
See also \cite{Bodirsky2021:datalog} for a semantic characterization of the intersection 
of MSO and Datalog in terms of infinite constraint satisfaction problems.
See also \cite{Rudolph2013:flag} where another formalism is introduced that
is contained in the intersection of MSO and Datalog. 

\paragr{Tree-shapedness}
 We say that a Datalog program $P$ is \emph{tree-shaped} if the
 incidence graph of each rule is acyclic. In 
 particular, this implies that no variable occurs twice in the 
 same conjunct in the rule body (but the rule head may contain 
 repeated occurrences of variables). Note that we do not require 
 the incidence graph of the rules to be connected, 
nor do we make any requirements 
 (say, in the case of binary relations) on the direction of edges. 

\begin{example}\label{ex:non-monadic}
Consider the tree-shaped Datalog program $P$ given by
the following two rules (where $\bfS_{in}$
consists of two binary relations, $E, F$):
\[R(x,y) \colondash E(x,u), F(u,y)  \qquad
R(x,y) \colondash E(x,u), R(u,v), F(v,y) 
\]
Then $R^{P(I)}$ contains all pairs $(a,b)$, such 
that there is a directed path from $a$ to $b$ in $I$
consisting of a number of $E$-edges followed by
an equal number of $F$-edges.

It follows from known facts about MSO (viz.~the
fact that MSO on words captures the regular languages)
that $(P,R)$ does not define an MSO query. In particular, $P$ is not
equivalent to a monadic Datalog program, or even an 
almost-monadic Datalog program.
\end{example} 

\vspace{-3mm}

\begin{restatable}{lemma}{lemtreeunfoldings}
\label{lem:tree-unfoldings}
Let $P$ be any tree-shaped Datalog program. Then, for each $R\in\bfS_{out}^P$, $\Unfoldings(P,R)$ consists of acyclic pointed instances.
\end{restatable}

\paragr{TAM Datalog}
A \emph{TAM Datalog program} is a tree-shaped, almost-monadic Datalog program. 
We will give a precise model-theoretic characterization of TAM Datalog
in terms of MSO.
We say that an MSO query $\phi(x_1, \ldots, x_k)$ is
 \emph{tree-determined} if for each pointed instance $(I,a_1, \ldots, a_k)$,
 we have that $I\models\phi(a_1, \ldots, a_k)$ if and only if there is an acyclic
 pointed instance $(J, b_1, \ldots, b_k)$ such that
 $J\models\phi(a_1, \ldots, a_k)$ and $(J,b_1, \ldots, b_k)\to (I,a_1, \ldots, a_k)$. Note that $J$ must be
 finite and that $J$ is not required to be connected.
 
\begin{restatable}{theorem}{thmmsototam}
\label{thm:mso-to-tam}
Let $\phi(x_1, \ldots, x_n)$ be an MSO formula over a schema $\bfS_{in}$. The following are equivalent:
\begin{enumerate}
    \item $\phi$ is definable by a TAM Datalog program,
    \item $\phi$ is definable by a tree-shaped Datalog program,
    \item $\phi$ is tree-determined.
\end{enumerate}
\end{restatable}

\begin{remark}
It is worth comparing this to the result in 
\cite{Gottlob2004:monadic} that states that monadic Datalog
and MSO have the same expressive power on finite trees. Besides
the fact that Theorem~\ref{thm:mso-to-tam} is a characterization
on arbitrary (finite) instances while the result in~\cite{Gottlob2004:monadic} is restricted to trees, 
there are a few other important differences: in~\cite{Gottlob2004:monadic},
it is assumed that the trees are represented as structures
in which the children of each node are ordered; that the 
signature
includes predicates marking the root, leafs, the first child 
of each node, and the last child of each node; and
 that each node of the tree is labeled by precisely 
 one of the (other) unary predicates in the signature. These assumptions together imply that every homomorphism between
 such trees is necessarily an isomorphism.
\end{remark}
\begin{restatable}[TAM Datalog is closed under composition]{corollary}{corclosureundercomposition}
For all TAM Datalog programs $P_1$ and
$P_2$  with $\bfS_{in}^{P_2}=\bfS_{out}^{P_1}$,
there is a TAM Datalog program $P_3=(\bfS_{in}^{P_1}, \bfS_{out}^{P_2}, \bfS'_{aux},\Sigma')$ such that, for all $\bfS_{in}^{P_1}$-instances $I$,
$P_3(I) = P_2(P_1(I))$.
\end{restatable}

We also provide a syntactic normal form for TAM Datalog programs.
A TAM Datalog program is \emph{simple} if every rule body
contains precisely one occurrence of a relation from $\bfS_{in}$.
For instance the program given in Example~\ref{ex:almost-monadic} 
is a simple TAM Datalog program. 

\begin{restatable}{theorem}{thmnormalform}
\label{thm:normal-form}
Every (connected) TAM Datalog program can be transformed in polynomial-time into an equivalent (connected) simple
TAM Datalog program.
\end{restatable}

\section{Right-Adjoints}
\label{sec:adjoints}

Recall that we can view an \edatalog program $P$ semantically as a mapping
from $\bfS_{in}^P$-instances to $\bfS_{out}^P$-instances that is
monotone with respect to homomorphisms: every
homomorphism $h:I\to I'$ gives rise 
to a homomorphism $h':P(I)\to P(I')$. To use the language of 
category theory, this means that $P$ is a \emph{functor from 
the category of $\bfS_{in}^P$-instances and homomorphisms to 
$\bfS_{out}^P$-instances and homomorphisms}.
Recall that, %
for functors $F:X\to Y$ and $G:Y\to X$ (where $X$ and $Y$ are
arbitrary categories), we say that 
$G$ is a \emph{right-adjoint} for $F$, and that $F$ is a 
\emph{left-adjoint} of $G$, if it holds that $F(I)\to J$ iff
$I\to G(J)$. 
\footnote{Note that, unlike in the standard category theoretical definition, we care only about existence of an arrow instead of a 1-to-1 correspondence between the two sets of arrows.}
In this section, we study the existence of 
right-adjoints for \edatalog
programs. 

\begin{example}
Consider the Datalog program $P=(\bfS_{in}, \bfS_{out}, \emptyset, \Sigma)$,
where $\bfS_{in} = \{R\}$, $\bfS_{out}=\{S\}$, and $\Sigma$ consists of the 
rules $S(x,y) \colondash R(x,y)$ and $S(x,y) \colondash R(y,x)$. 
We can think of the input instances for $P$ as directed graphs, and we can think
of $P(I)$ as the \emph{symmetric closure} of $I$.
For every $\bfS_{out}^P$-instance
$J$, let $\Omega(J)$ be the $\bfS_{in}^P$-instance that is the 
\emph{maximal symmetric sub-instance} of $J$, that is,
$\Omega(J)$ consists of all facts $R(x,y)$ for which it holds that
$J$ contains both $S(x,y)$ and $S(y,x)$. 
It is not hard to see that $\Omega$ is a right-adjoint of $P$. That is, $P(I)\to J$  iff
$I\to \Omega(J)$.
\end{example}

\begin{example}
Consider the Datalog program $P=(\bfS_{in}, \bfS_{out}, \emptyset, \Sigma)$,
where $\bfS_{in} = \{Q_1,Q_2\}$, $\bfS_{out}=\{Q_3\}$, and $\Sigma$ consists of the 
rule
$Q_3() \colondash Q_1(x), Q_2(y)$.
This Datalog program does \emph{not} have a right-adjoint in the above sense.
Indeed, let $J$ be the empty instance. Then $P(I)\to J$ holds if and only
if either $I$ has no $Q_1$-facts or $I$ has no $Q_2$-facts, a condition
that cannot be equivalently characterized by the existence of a homomorphism
from $I$ to any fixed single instance $J'$. However, it can be shown
that $P(I)\to J$ if and only if either $I\to J'_!$ or $I\to J'_2$, where
$J'_1=\{Q_1(a)\}$ and $J'_2=\{Q_2(a)\}$. If we generalize the notion of right-adjoint by allowing $\Omega(J)$ to be a finite set of instances, then, as we will see later, $P$ \emph{does} admit
a right-adjoint. As we will see later, the fact that $\Omega(J)$  consists of multiple instances, is related to the fact that
the program includes a rule whose incidence graph is not connected.
\end{example}

Motivated by the above examples and other considerations that will
become clear in Section~\ref{sec:dualities}, the precise notion of right-adjoints that we will adopt here is a little more refined:

\begin{definition}[Generalized Right-Adjoints] \label{def:adjoints}
Let $P$
 be any \edatalog program.
 By a \emph{generalized right-adjoint} for $P$ we will mean a 
function $\Omega_P$ that maps every $J\in \Inst[\bfS_{out}^P]$ to
a finite set of pairs $(J',\iota)$ with
  $J'\in \Inst[\bfS_{in}^P]$ and $\iota:adom(J')\hookrightarrow adom(J)$ a partial function,
  such that, for all
  $I\in \Instinf[\bfS_{in}^P]$, $P(I)\to J$ iff $I\to J'$ for some $(J',\iota)\in\Omega_P(J)$, in such a way that the 
  following diagram commutes:
  \[
  \begin{tikzcd}
    P(I) \arrow{r}{} & J  \\%
    I \arrow{r}{} \arrow[hook]{u}{id}    & J'\arrow[swap,hook]{u}{\iota}
  \end{tikzcd}
  \]
\end{definition}

Here, the notation $f:X\hookrightarrow Y$ indicates 
that $f$ is a partial function from $X$ to $Y$.

Recall that, when $P$ is an arbitrary $\edatalog$-program, $P(I)$ may be an infinite instance, even when $I$ is finite. It is for this reason
(and because, later on, we will consider compositions of \edatalog programs)
that the above definition requires the adjoint operator to behave well
even for infinite instances $I$.
Incidentally, the results in our paper 
do not depend on the fact that $J$ is finite, and we could change the above definition such that $J,J'\in\Instinf[\bfS^P_{out}]$, but we have opted not to do so since this is not necessary for our use cases.

This notion of generalized right-adjoint is properly defined and behaves as one expects:

\begin{restatable}{theorem}{thmcomposition}
\label{thm:composition}
If two \edatalog programs have generalized right-adjoints, then so does their composition.
\end{restatable}

Our main results, in this section, will show that certain classes of  \edatalog programs admit generalized right-adjoints.
In the next sections, we will apply this to obtain  new results, such as regarding
the existence of homomorphism dualities.

\begin{remark}
While we are specifically interested in right-adjoints in this paper, one may
also wonder what it means for a \edatalog program to admit a (generalized) left-adjoint. Generalized left-adjoints for \edatalog programs are closely related to query rewritings, as studied in the literature on data integration and 
data exchange. A Datalog program $P$ has
a generalized left adjoint iff $P$ is equivalent to a non-recursive
Datalog program. Indeed, if $P$ has a generalized left-adjoint $\Theta$, then,
for each $R\in \bfS_{out}^P$, the $\bfS_{in}^P$-instances in 
$\Theta(\{R(a_1, \ldots, a_n)\})$ correspond  to the members of $\Unfoldings(P,R)$
(cf.~\cite{Foniok2015:functors,DalmauKO}).
\end{remark}

\begin{restatable}{theorem}{thmtamadjoint}
\label{thm:tam-adjoint}
Every TAM Datalog program $P$ has a generalized right-adjoint $\Omega_P$. If $P$ is connected then $\Omega_P(J)$ is always a singleton.
Moreover, $\Omega_P(J)$ is computable in 2Exptime given $J$ and $P$, and in ExpTime whenever the arity of $P$ is bounded.
\end{restatable}

The proof of Theorem~\ref{thm:tam-adjoint} (which is given in the appendix) makes crucial use of both the tree-shapedness and the almost-monadicity of the 
Datalog program. Indeed, both properties are
important for the existence of generalized right-adjoints:

\begin{restatable}{proposition}{proptreeshapednoadjoint}
\label{prop:treeshaped-no-adjoint}
The tree-shaped Datalog program $P$ in Example~\ref{ex:non-monadic} (which is not almost-monadic) does
not admit a generalized right-adjoint. 
\end{restatable}

\vspace{-2mm}

\begin{restatable}{proposition}{propcacyclicnoadjoint}
\label{prop:cacyclic-no-adjoint}
The monadic Datalog program given by the single rule
$Ans(x) \colondash E(x,x)$
does not admit a generalized right-adjoint. 
\end{restatable}

For Boolean non-recursive programs, there is a converse to Theorem~\ref{thm:tam-adjoint}
(cf.~also~\cite[Theorem 2.5]{Foniok2015:functors}). 
The following theorem follows from results that we will prove in Section~\ref{sec:dualities} 
(specifically, Theorem~\ref{thm:adjoint-dualities}
in combination with Theorem~\ref{thm:cacyclic-duals}):

\begin{restatable}{theorem}{thmnonrecursive}
\label{thm:non-recursive}
For Boolean non-recursive Datalog programs $P$, the following are
equivalent:
\begin{enumerate}
    \item $P$ admits a generalized right-adjoint,
    \item $P$ is equivalent to a TAM Datalog program.
\end{enumerate}
\end{restatable}

\medskip

Next, we identify a second class of \edatalog programs that admits
generalized right-adjoints.
We say that a \edatalog program is \emph{strongly linear} if
the body of each rule consists of a single
atom (over $\bfS_{in}\cup\bfS_{aux}$) without repeating variables. 
Strongly linear
\edatalog rules are also known as \emph{LAV (``Local-as-View'') constraints}.
Note that strong linearity is a stronger restriction 
than the mere requirement that every rule body contains at most one $\bfS_{aux}$-atom (which is often called linearity).
For instance the program in Example~\ref{ex:tc}
is not strongly linear. Every strongly linear program is clearly tree-shaped.

\begin{example}
Consider the \edatalog program $P=(\bfS_{in}, \bfS_{out}, \bfS_{aux}, \Sigma)$,
where $\bfS_{in}=\{R_{in}\}$, $\bfS_{out}=\{R_{out}\}$, $\bfS_{aux}=\{R\}$, and $\Sigma$ consists of 
the three rules
\[R(x,y,z) \colondash R_{in}(x,y,z) ~~~~~~~~
\exists uv~ R(y,u,v) \colondash R(x,y,z) ~~~~~~~~
R_{out}(x,y,z) \colondash R(x,y,z) 
\]

An $\bfS_{in}$-instance is a database instance consisting
of a single ternary relation, and we can think of the $P$ as performing two
things: copying the input data to the output, and adding additional facts to make sure that the
inclusion dependency $\forall xyz(R(x,y,z)\to \exists uv (R(y,u,v)))$ is satisfied. 
This program $P$ is strongly linear. Note that $P$ is recursive and not 
weakly acyclic. Indeed, $P(I)$ can be an infinite instance even when $I$ is finite.
\end{example}

\vspace{-3mm}

\begin{restatable}{theorem}{thmlinearadjoints}
Every strongly linear \edatalog
program $P$ has a generalized right-adjoint $\Omega_P$. Moreover, for each 
$\bfS_{out}$-instance $J$, $\Omega_P(J)$ is a singleton set and can be computed
in ExpTime from $J$ and $P$. If $P$ is fixed $\Omega_P(J)$ can be computed in polynomial time. %
\end{restatable}

Right adjoints for strongly linear \edatalog programs consisting of a single non-recursive rule were initially given, using a different terminology, in \cite{Pultr70}.

Based on the above results, we can define a larger fragment of \edatalog that admits generalized right-adjoints: every \edatalog program that can be 
represented as a finite composition (in any order) of TAM Datalog programs and 
strongly linear \edatalog programs admits a generalized right-adjoint. It is possible to give a syntactic definition
of this language (using stratification) but we will  not do so here.

\section{Intermezzo: TGDs}

A \emph{tuple-generating dependency (TGD)} is a first-order sentence
of the form 
\[ \forall\textbf{x}(\phi(\textbf{x})\to \exists \textbf{y}\psi(\textbf{x},\textbf{y}))\]
where $\phi(\textbf{x})$ and $\psi(\textbf{x},\textbf{y})$ are  conjunctions
of relational atomic formulas. TGDs allow expressing a wide variety of constraints, including database dependencies such as inclusion dependencies. 

Every finite set of TGDs naturally gives rise to an \edatalog program.
More precisely, for any set $\Sigma$ of TGDs over a schema $\bfS$, 
we will denote by $P_\Sigma$ the $\edatalog$ program 
with $\bfS_{in}^P=\{R_{in}\mid R\in\bfS\}$, $\bfS_{out}^P=\{R_{out}\mid R\in\bfS\}$, and $\bfS_{aux}^P=\bfS$,
consisting
of the TGDs in $\Sigma$ as \edatalog rules
(where $\forall\textbf{x}(\phi(\textbf{x})\to \exists \textbf{y}\psi(\textbf{x},\textbf{y}))$ becomes $\exists\textbf{y}\psi(\textbf{x},\textbf{y}) \colondash \phi(\textbf{x})$), plus the ``copy constraints'' $R(\textbf{x}) \colondash R_{in}(\textbf{x})$
and $R_{out}(\textbf{x})\colondash R(\textbf{x})$
for each $R\in \bfS$. 

Although the input and output schemas of $P_\Sigma$ are renamings of $\bfS$, we will be sloppy and write
$P(I)$ even when $I$ is an $\bfS$-instance, with the understanding that
relation symbols are renamed in the obvious way; and similarly, we
will treat $P(I)$ as an $\bfS$-instance.

The \edatalog program $P_\Sigma$ ``captures'' $\Sigma$ in the 
following sense: 

\begin{lemma} 
\label{lem:capturing}
Let $\Sigma$ be any finite set of TGDs. Then:
\begin{enumerate}
\item For all $\bfS$-instances $I$, $I\subseteq P_\Sigma(I)$ and $P_\Sigma(I)\models\Sigma$.
\item For all $\bfS$-instances $I\models\Sigma$, $P_\Sigma(I)\to_{adom(I)} I$.
\end{enumerate}
\end{lemma}

For any property $X$ of \edatalog programs
(e.g., \emph{tree-shaped}, \emph{weakly acyclic}, or \emph{having a generalized right-adjoint}), we will say that a finite set of TGDs
$\Sigma$ has property $X$ if $P_\Sigma$ does.

Next, we will give some examples of sets of TGDs that have a
generalized right-adjoint.
We will see the importance of this property in the next sections.

\begin{example}
\label{ex:tgd-adjoint}
The following sets of TGDs have a generalized right-adjoint:

\begin{itemize}
\item $\Sigma_1=\{\forall xyz(R(x,y)\land R(y,z)\to R(x,z))\}$. 
   To see that $P_{\Sigma_1}$ has a generalized right-adjoint, observe
   that it consists of the rules depicted on the left:

   \hbox{\parbox{.5\textwidth}{
   \[\begin{array}{lll}
         R(x,y) &\colondash& R_{in}(x,y) \\
         R(x,z) &\colondash& R(x,y), R(y,z) \\
         R_{out}(x,y) &\colondash& R(x,y) 
   \end{array}\]}%
   \parbox{.5\textwidth}{
   \[\begin{array}{lll}
         R(x,y) &\colondash& R_{in}(x,y) \\
         R(x,z) &\colondash& R(x,y), R_{in}(y,z) \\
         R_{out}(x,y) &\colondash& R(x,y) 
   \end{array}\]}}
   $P_{\Sigma_1}$ is neither a TAM Datalog program, nor a strongly linear \edatalog program. However, it is equivalent to the program $P'$ consisting of
   the rules depicted on the right.
   Note how we have replaced one occurrence of $R$ by $R_{in}$. The equivalence of $P_{\Sigma_1}$ and $P'$ is easy to show.
   Furthermore, $P'$ is a TAM Datalog program (where the articulation position 
   of $R$ is the second position). Since $P_{\Sigma_1}$ is equivalent
   to $P'$ and $P'$ has a generalized right-adjoint, $P_{\Sigma_1}$ does too
   (indeed, it has the same generalized right-adjoint). 
\item $\Sigma_2 = \{\forall xy(R(x,y)\to \exists z(R(y,z))\}$. To see that
$P_{\Sigma_2}$ has a generalized-right adjoint it suffices to observe that
it is a strongly linear \edatalog program. Indeed, the same applies to any
set of $\Sigma$ consisting only of inclusion dependencies. 
\item $\Sigma_3 = \Sigma_1\cup\Sigma_2$.
   Although $\Sigma_3$ is neither strongly linear, nor equivalent to a TAM Datalog program, it is equivalent to the composition of two \edatalog 
   programs, namely $P_{\Sigma_2}$ and $P_{\Sigma_1}$. To see this, note 
   that whenever $I\models\Sigma_2$, then also $P_{\Sigma_1}(I)\models\Sigma_2$. 
   (Coincidentally, the order in which we perform the composition here matters:
   if $I\models\Sigma_1$, it does not, in general, follow that $P_{\Sigma_2}(I)\models\Sigma_1$!)
   Since $P_{\Sigma_1}$ and $P_{\Sigma_2}$ each have a 
    generalized right-adjoint, their composition does too. 
\item $\Sigma_4 = \{\forall xyzu (R(x,y)\wedge R(y,z)\wedge R(z,u)\rightarrow R(x,u)), \forall xy(R(x,y)\rightarrow R(y,x))\}$. Just as in the case of $\Sigma_1$, we have
that, although $P_{\Sigma_4}$ is not a TAM Datalog program, it can be easily rewritten as one. We will return to this example later, in Remark~\ref{rem:tgd-adjoint-acyclic}.
\item Let us say
that a TGD is \emph{monadic} if the 
relation in the rule head is monadic.
Then, for every finite set $\Sigma$
of monadic tree-shaped TGDs, $P_\Sigma$ 
is a TAM Datalog program.
\end{itemize}
\end{example}

\begin{remark}
Most of the above examples  involve
adhoc arguments. 
We leave it as an open problem to define a large syntactic class of (sets of) TGDs that have a generalized
right-adjoint, which includes $\Sigma_1$. Theorem~\ref{thm:tam-adjoint} with Theorem~\ref{thm:mso-to-tam}
does imply that, for finite sets of tree-shaped
TGDs $\Sigma$, if $P_\Sigma$ is MSO-definable then
$\Sigma$ has a generalized right-adjoint.
\end{remark}

\section{Homomorphism Dualities}
\label{sec:dualities}

For any set of instances $X$, let $X\upclosure = \{A\mid B\to A$ for some $B\in X\}$, and let $X\downclosure = \{A\mid A\to B$ for some $B\in X\}$.
A \emph{homomorphism duality} is a pair of sets of instances $(F,D)$, 
such that $F\upclosure$ is the complement of $D\downclosure$. The same
definition extends naturally to pointed instances. By a \emph{finite}
homomorphism duality, we mean a homomorphism duality $(F,D)$ where $F$
and $D$ are finite sets. By a \emph{tree duality}, we mean a 
homomorphism duality $(F,D)$ where $F$ is a (possibly infinite)
set of (not-necessarily-connected) acyclic instances, and $D$ is finite. 

The study of dualities originated in combinatorics (see \cite{HellNesetril2004}) motivated by its links to the structure of the homomorphism partial order, and the complexity of deciding the existence of homomorphism between graphs and, more generally, relational structures (a.k.a.~constraint satisfaction problems or CSPs).
Indeed, dualities have played an important role in the study of CSPs. In particular, it was shown \cite{Atserias08} that the CSPs definable in FO are precisely those whose template are the right-hand side of a finite duality. In a similar vein, %
the CSPs solvable by the well-known \emph{arc-consistency} algorithm are precisely those whose template is the right-hand side of a tree duality. More generally, 
the CSPs that are solvable by local consistency methods are
those whose template is the right-hand side of 
a homomorphism duality whose left-hand side consists of instances of bounded treewidth. See \cite{BulatovKL08} for a survey on the connections between duality and consistency algorithms.

\begin{example} 
Let $\bfS=\{R\}$, where $R$ is a binary relation symbol, and let  $n\geq 1$.
Let $L_n$ be the finite linear order of length $n$, and 
let $P_{n+1}$ be the directed path of length $n+1$.
Then $(\{P_{n+1}\},\{L_n\})$ is a finite homomorphism duality.
\end{example}

\begin{example}
Let $\bfS=\{P_0,P_1,E\}$, where $P_0$ and $P_1$ are unary and $E$ is binary, and consider the two-element $\bfS$-instance $I=\{P_0(0), P_1(1), E(0,0), E(1,1)\}$ (without distinguished elements). It is easy to see that, for all $\bfS$-instances $J$, $J\to I$ holds if and only if 
no connected component of $J$ contains both a $P_0$-fact
and a $P_1$-fact. This can be 
expressed in the form of a tree duality:
let $F$ be the set of all (acyclic) instances consisting of an oriented path that connects
 a $P_0$-node to a $P_1$-node.
Then
$(F,\{I\})$ is a homomorphism duality.
\end{example}

\begin{theorem}[\cite{FoniokNT08,tCD2022:conjunctive}]
\label{thm:cacyclic-duals}
Fix a schema $\bfS$ and $k\geq 0$.
Let $F$ be any finite set of pairwise homomorphically incomparable $k$-ary pointed instances over $\bfS$. 
The following are equivalent:
\begin{enumerate}
\item There is a finite set of $k$-ary pointed instances $D$ over $\bfS$ such that $(F,D)$ is a 
homomorphism duality.
\item Each pointed instance in $F$ is homomorphically equivalent to a c-acyclic pointed instance. 
\end{enumerate}
Moreover (for fixed $\bfS$ and $k$), given a set $F$ of c-acyclic  
pointed instances, such a set $D$ can be
computed in ExpTime.
\end{theorem}

The ExpTime bound is not explicitly stated in \cite{tCD2022:conjunctive} but follows from  results
in that paper.

\paragr{Constructing Dualities through Adjoints}
The following theorem establishes a close relationship between 
generalized right-adjoints and homomorphism dualities. Specifically,
it shows that generalized right-adjoints can be used to construct
duals. This approach was first used in \cite{Foniok2015:functors} and \cite{DalmauKO}, where right-adjoints are applied to derive the dual of 
a tree.

\begin{restatable}{theorem}{thmadjointdualities}
\label{thm:adjoint-dualities}
Let $P$ be any Datalog program that has a generalized right-adjoint.
Then, for each $R\in\bfS_{out}^P$, there is a
finite set of pointed $\bfS_{in}$-instances $D$ such that $(\Unfoldings(P,R),D)$ is a homomorphism 
duality. 
\end{restatable}

\begin{proof} 
We may assume without loss of generality
that $\bfS_{out}^P=\{R\}$.
Let $J$ be the 
$\bfS_{out}^P$-instance
with $adom(J)=\{b_1, \ldots, b_k, c\}$ 
(for $k=arity(R)$)
containing all 
$R$-facts over $adom(J)$ except $R(b_1, \ldots, b_k)$.
Let 
$D=\{(J',\textbf{b}')\mid (J',\iota)\in\Omega_P(J), \textbf{b}'\in adom(J')^k, \iota(\textbf{b}')=\textbf{b}\}$, where $\textbf{b}=b_1, \ldots, b_k$.
We claim that $(\Unfoldings(P,R), D)$ is a homomorphism duality.
Let $(C,\textbf{c})$ be any $\bfS_{in}$-instance with $k$ distinguished elements. Then an instance in $\Unfoldings(P,R)$ homomorphically maps to $(C,\textbf{c})$ iff
$R(\textbf{c}) \in P(C)$ iff  
$(P(C),\textbf{c})\not\to(J,\textbf{b})$ iff (by the adjoint property)
$(C,\textbf{c})\not\to (J',\textbf{b}')$ for all $(J',\iota)\in \Omega_P(J)$ and $\textbf{b}'$ with $\iota(\textbf{b}')=\textbf{b}$.
\end{proof}

In particular, this implies that, 
for every TAM Datalog program $P$ and for each $R\in\bfS_{out}^P$, there is a
finite set of pointed instances $D$ such that $(\Unfoldings(P,R),D)$ is a homomorphism duality.
Since $\Unfoldings(P,R)$ consists of acyclic instances whenever $P$ is a TAM Datalog program, this gives us a systematic way of constructing tree-dualities.

Observe that, in the special case where $P$ is a TAM Datalog program, or a strongly linear $\edatalog$ program, the proof of Theorem~\ref{thm:adjoint-dualities} yields a ExpTime algorithm for computing $D$ from $(P,R)$ provided the arity of the relations in $P$ is bounded. %

In fact, the following theorem says that every tree-duality can be 
obtained in this way.

\begin{corollary}
\label{co:homdual}
Let $F$ be any set of acyclic pointed instances.
 The following 
are equivalent:
\begin{enumerate}
    \item There is a finite set of pointed instances $D$ such that $(F,D)$ is a homomorphism duality
    \item $F\upclosure=\Unfoldings(P,R)\upclosure$ for some TAM Datalog program $P$ and $R\in\bfS_{out}^P$.
\end{enumerate}
\end{corollary}

\begin{proof}
From 1 to 2: It is well known that, for any finite set of pointed instances $D$, 
there is an MSO formula $\phi$ that defines $D\downclosure$. Hence, by duality, $\neg\phi$
defines $F\upclosure$.
Furthermore, the fact
that $F$ consists of acyclic pointed instances 
implies that $\neg\phi$ is tree-determined. Therefore, the direction 1 to 2 follows
 from Theorem~\ref{thm:mso-to-tam}.  
The direction from 2 to 1 follows from Theorem~\ref{thm:adjoint-dualities}.
\looseness=-1
\end{proof}

It is possible to strengthen Corollary \ref{co:homdual} by showing that the above conditions (1) and (2) are, in turn, equivalent to the fact that $F\upclosure=G\upclosure$ for some regular set $G$ of acyclic queries (where ``regular'' needs to be defined in a suitable way, as in \cite{ErdosPTT17}). This follows from the fact that Theorem 3.9 uses tree-automata as an intermediate step in the proof. We note that the special case of this equivalence for  Boolean CQs over digraphs was proven in \cite{ErdosPTT17}.

 Theorem~\ref{thm:adjoint-dualities} also implies that every finite set of acyclic pointed instances $F$ is
the left-hand side of a finite homomorphism duality: it suffices to let $P$ be the  TAM Datalog program containing 
 a single non-recursive rule for each $(I,\textbf{a})\in F$, whose canonical instance is $(I,\textbf{a})$. 
Then, the unfoldings of $P$ are, up to isomorphism, precisely
the  pointed instances in $F$.
It follows from Theorem~\ref{thm:adjoint-dualities} that there is a finite set $D$
such that $(F,D)$ is a homomorphism duality.
\looseness=-1

\begin{remark}\label{rem:cacyclic}
Recall that a pointed instance is the left-hand side of a finite homomorphism duality if and only if (up to homomorphic equivalence) it is 
c-acyclic (Theorem~\ref{thm:cacyclic-duals}). In the light of
this, it is natural to ask whether the above ``dualities through adjoints'' 
technique can be used to construct a finite homomorphism duality for any c-acyclic pointed instance. Proposition~\ref{prop:cacyclic-no-adjoint} shows
that this is not the possible. Note that the
canonical instance of the rule of the program in 
Proposition~\ref{prop:cacyclic-no-adjoint}
is
$(\{E(a,a)\},a)$, which is c-acyclic (but not acyclic).
\end{remark}

\paragr{Homomorphism dualities relative to a background theory}
In many settings, one is interested in restricted classes of relational structures, such as
 linear orders, equivalence relations, database instances that satisfy given integrity constraints, models of a knowledge base, etc. 
In this section, we look at the question when homomorphism dualities exist in  the category of all instances that satisfy a given background theory. 

A few results are known.
An undirected graph can be viewed as an instance over a 
schema $\bfS$ consisting of a single binary relation symbol $E$, 
satisfying the TGDs $\forall xy(E(x,y)\to E(y,x))$ and $\forall x\neg E(x,x)$. It is known that the category of
undirected graphs and homomorphisms has no finite dualities, up to homomorphic equivalence, other than the trivial duality $(\{K_2\}, \{K_1\})$, where $K_1$ and $K_2$ are the 2-element clique and  the empty graph, respectively~(cf.~\cite{HellNesetril2004}). Similarly, a finite algebra of a similarity type $\sigma$ can be viewed as an $\bfS$-instance, with $\bfS=\{R_f\mid f\in\sigma\}$
satisfying  $\Sigma=\{\forall\textbf{x}\exists y R_f(\textbf{x},y), \forall \textbf{x}yz (R_f(\textbf{x},y)\land R_f(\textbf{x},z)\to y=z)\mid f\in\sigma\}$, 
and it is known that, in the category of finite algebras,
no non-trivial finite
dualities exist~\cite{Ball2010}.

The next result shows how to obtain finite homomorphism dualities from classes  
definable by a weakly acyclic set of TGDs that have a generalized right-adjoint.

\begin{restatable}{theorem}{thmcacyclicdualititestgd}
\label{thm:c-acyclic-dualities-tgd}
Let $\Sigma$ be a finite weakly acyclic set of TGDs that has a generalized right-adjoint. 
Let $F$ be any finite set of pointed instances.
If each member of $F$ is of the form 
$(P_\Sigma(A),\textbf{a})$ for some c-acyclic pointed
instance $(A,\textbf{a})$, then 
$F$ has finite duality w.r.t. $\Sigma$.
\end{restatable}

Regarding complexity, consider the case where $\Sigma$ is a
fixed set of TGDs (not treated as part of the input) such
that $P_\Sigma$
is equivalent to a TAM Datalog program.
Then Theorem~\ref{thm:c-acyclic-dualities-tgd}
yields a 2ExpTime algorithm for computing the dual set
$D$ from $F$, assuming $F$ is specified by
the underlying set of c-acyclic structures $(A,\textbf{a})$. It follows that, for instance,
for the class of transitive digraphs (which,
as we saw earlier, is captured by a TAM Datalog program),
we have a 2ExpTime-algorithm for constructing duals
for digraphs that are specified as the transitive closure of an
acyclic digraph. The same holds when $\Sigma$
is a weakly acyclic set of inclusion 
dependencies, 
or, more generally, when $P_\Sigma$ is
equivalent to a weakly acyclic strongly linear Datalog program.
In fact, in this case, we get an ExpTime upper bound.

For the special
case of monadic tree-shaped TGDs, 
we can prove a converse:

\begin{restatable}{theorem}{thmmonadicduality}
\label{thm:monadic-duality}
Let $\Sigma$ be any set of tree-shaped monadic TGDs.
Let $F$ be any finite set of pairwise homomorphically-incomparable pointed instances
$(A,\textbf{a})$ with $A\models\Sigma$. 
Then, the following are equivalent:
\begin{enumerate}
    \item $F$ has finite duality w.r.t. $\Sigma$, 
    \item Each $(A,\textbf{a})\in F$ is homomorphically equivalent to 
        $(P_\Sigma(A'),\textbf{a})$ for some c-acyclic $(A',\textbf{a})$. 
\end{enumerate}
\end{restatable}

\begin{remark}\label{rem:tgd-adjoint-acyclic}
Theorem~\ref{thm:monadic-duality} cannot be lifted to arbitrary finite weakly acyclic sets of TGDs
that admit a generalized right-adjoint.
Consider the weakly acyclic set of TGDs 
(over a schema with a single binary relation)
$\Sigma = \{\forall xyzu (R(x,y)\wedge R(y,z)\wedge R(z,u)\rightarrow R(x,u)), \forall xy(R(x,y)\rightarrow R(y,x))\}$. 
As we pointed out in Example~\ref{ex:tgd-adjoint}, $P_\Sigma$
is equivalent to a TAM Datalog program, and hence
$\Sigma$ has a generalized right-adjoint. 
Let $A$ be the instance (without 
distinguished elements) $\{R(a,a)\}$, and let 
$B$ be the instance $\{R(a,b),R(b,a)\}$.
Then $(\{A\},\{B\})$ is a 
homomorphism duality w.r.t.~$\Sigma$.
Indeed, let $C$ be an instance satisfying $\Sigma$ and assume that $A\not\rightarrow C$ (i,e, $C$ has no loop). Since $C$ satisfies $\Sigma$ it follows that $C$ has no odd cycle and, hence, is homomorphic to $B$. However, it is easy to see that every instance $A'$ satisfying $P_{\Sigma}(A')=A$ must have a cycle.
\end{remark}

\paragr{Dualities in ABox categories}
Fix some finite set $\Sigma$ of TGDs.
An \emph{ABox} is, intuitively, a finite database instance $I$ that
is treated as an incomplete database and whose completion is  
$P_\Sigma(I)$ (cf.~\cite{baader2003basic}).
We do not assume that $\Sigma$ is weakly acyclic. Therefore, we can think of an ABox $I$
as a finite representation of a possibly infinite instance $P_\Sigma(I)$. 

While, at the level of its specification, an ABox is nothing
else than a finite instance, it naturally comes with a different
type of morphism, capturing the intended semantics: 
we write $h:I\to_\Sigma J$ if $h$ is a partial function from 
$adom(I)$ to $adom(J)$ which can be extended to a homomorphism $h':P_\Sigma(I)\to P_\Sigma(J)$. 

We denote by $ABox_\Sigma[\bfS]$ the category of ABoxes, where the 
arrows are the $\to_\Sigma$-morphisms as described above.
Our interest in this category comes from the fact that
it plays a fundamental role in knowledge representation,
and more specifically, description logic. 
The proof of Theorem~\ref{thm:c-acyclic-dualities-tgd} also 
implies the following result for dualities in the
category $ABox_\Sigma[\bfS]$ (with similar 
complexity bounds):

\begin{restatable}{theorem}{thmaboxduality}
\label{thm:abox-duality}
Let $\Sigma$ be any finite set of TGDs that has a generalized right-adjoint. 
Every finite set $F$ of c-acyclic pointed ABoxes has
finite duality within the 
      category $ABox_\Sigma[\bfS]$.
\end{restatable}

\section{Application: Uniquely Characterizable UCQs}
\label{sec:applications}

In this section, we show-case one application of our
results on homomorphism dualities. It was shown in~\cite{FoniokNT08,tCD2022:conjunctive} 
that every c-acyclic union of conjunctive queries (UCQ) is uniquely characterizable
by a finite collection of labeled examples.
In fact, a UCQ $q$ is uniquely characterizable by a 
finite collection of labeled examples, if and only if
$q$ is equivalent to a c-acyclic UCQ. 
In this section, we study whether the same holds 
over restricted
classes definable by a finite weakly acyclic set of TGDs.

We assume the reader is familiar with the definition
of UCQs (cf.~Appendix~\ref{app:more-prels}). We call a 
UCQ $q$ \emph{c-acyclic} if the (pointed) canonical instance
of each CQ in $q$ is c-acyclic. 

Let $\bfS$ be a schema, and $\Sigma$ a first-order theory
over $\bfS$. 
We say that two UCQs $q,q'$ (over schema $\bfS$) are \emph{equivalent w.r.t.~$\Sigma$} if 
for all $\in\Inst[\bfS]$ with $I\models\Sigma$, $q(I)=q'(I)$. 
By a \emph{collection of labeled examples} for a $k$-ary
query, we mean a pair $(E^+,E^-)$ of finite sets of 
pointed instances with $k$ distinguished elements. 
A UCQ $q$ \emph{fits} such $(E^+,E^-)$ if $\textbf{a}\in q(A)$ for all $(A,\textbf{a})\in E^+$, and 
$\textbf{a}\not\in q(A)$ for all $(A,\textbf{a})\in E^-$. 
A collection of labeled examples $(E^+,E^-)$ 
\emph{uniquely characterizes} a UCQ $q$ w.r.t.~$\Sigma$,
if $q$ fits $(E^+,E^-)$, and every UCQ
that fits $(E^+,E^-)$ is equivalent 
to $q$ w.r.t.~$\Sigma$.

\begin{restatable}{theorem}{thmuniqchar}
\label{thm:uniq-char}
Let $\bfS$ be a schema and $\Sigma$ a finite weakly acyclic set of TGDs that
has a generalized right-adjoint.
Then every c-acyclic UCQ $q$ over $\bfS$ is uniquely characterized w.r.t.~$\Sigma$ by a finite collection
      of labeled examples satisfying $\Sigma$.
\end{restatable}

Motivated by use cases in knowledge representation (cf.~\cite{Funk2022:frontiers})
, we also present a variant of Theorem~\ref{thm:uniq-char} for \emph{ABox-examples}. Recall that,
in the presence of a set of TGDs $\Sigma$,
an ABox is a (finite) instance that does not necessarily 
satisfy $\Sigma$, but that is treated as a finite specification
of the possible-infinite instance $P_\Sigma(I)$. From this 
perspective,
an \emph{ABox-example} is simply a pointed instance $(A,\textbf{a})$.
We say that a UCQ $q$ \emph{fits} an collection of labeled ABox-examples $(E^+,E^-)$ 
w.r.t.~$\Sigma$, if $\textbf{a}\in q(P_\Sigma(A))$ for all
$(A,\textbf{a})\in E^+$, and  $\textbf{a}\not\in q(P_\Sigma(A))$ for all $(A,\textbf{a})\in E^-$. A collection of labeled ABox-examples $(E^+,E^-)$ \emph{uniquely characterizes} a UCQ $q$ if $q$ fits $(E^+,E^-)$ and every UCQ $q'$ that fits
$(E^+,E^-)$, is equivalent to $q$ w.r.t. $\Sigma$.

The phrase ``equivalent w.r.t.~$\Sigma$'' here requires some further
discussion, since it was defined earlier in reference to 
finite instances only, while in the present context, it 
is more natural to consider finite and infinite instances, or at least all instances
that are of the form $P_\Sigma(I)$ for some finite instance $I$. To simplify the picture and avoid confusion, we will restrict attention to
sets of TGDs $\Sigma$ that have the \emph{finite
controllability} property, meaning that for all UCQs $q, q'$, it holds that
$q$ and $q'$ are equivalent over finite instances satisfying $\Sigma$
if and only if $q$ and $q'$ are equivalent over all
(finite and infinite) instances satisfying $\Sigma$.
This implies that the aforementioned three notions of
equivalence all coincide. It is known, for instance, 
that all finite
sets of strongly linear TGDs (in particular, 
inclusion dependencies) are finitely controllable~\cite{Rosati11:finite,Barany2013querying}.
\looseness=-1

\begin{restatable}{theorem}{thmuniqcharabox}
Let $\bfS$ be a schema and $\Sigma$ a finite set of TGDs that
has a generalized right-adjoint and that is finitely controllable. 
Then every c-acyclic UCQ $q$ over $\bfS$ is uniquely characterized w.r.t.~$\Sigma$ by a finite
     collection of labeled ABox-examples.
\end{restatable}

Note that this result, unlike Theorem~\ref{thm:uniq-char}, does not require weak acyclicity. In particular, it applies
whenever $\Sigma$ is a set of strongly linear TGDs.

\section{Conclusion}

We introduced a new fragment of Datalog, TAM Datalog,
that is semantically well-behaved (closed 
under composition and having a natural semantic characterization) and admits generalized right-adjoints. We also showed that strongly linear \edatalog programs admit generalized right-adjoints. 
We used these results to obtain new methods for 
constructing
homomorphism dualities w.r.t.~a background theory, 
and, subsequently, for constructing unique characterizations for UCQs w.r.t.~a background theory (addressing an open question
from~\cite{tCD2022:conjunctive}).
\looseness=-1

To illustrate the latter,
consider the class $C$ of transitive digraphs. It 
follows from our results that each c-acyclic UCQ $q$
admits a unique characterization w.r.t.~$C$, that is, a 
finite set of labeled examples from $C$
and that uniquely characterize $q$ w.r.t.~$C$. Moreover, the examples in question can be constructed in 2ExpTime. Similarly, if we consider the class of digraphs
satisfying the inclusion dependency $\forall xy(R(x,y)\to\exists z R(y,z))$, then every 
c-acyclic UCQ admits a uniquely characterizing set of
ABox-examples which can be computed in ExpTime.
\looseness=-1

We leave as open problems for future research:
(i)
identifying a syntactic criterion that guarantees
that a given set of TGDs admits a generalized right-adjoint;
(ii)  
extending our results to the case with equality-generating
dependencies; (iii) obtaining tight complexity bounds 
for the task of constructing homomorphism dualities from
TAM Datalog programs.

\bibliographystyle{plainurl}
\bibliography{bib}

\appendix

\section{More preliminaries}
\label{app:more-prels}

\paragr{Weak acyclicity}
The \emph{dependency graph} of an \edatalog program
is a directed graph that has as its nodes all
pairs $(R,i)$ where $R\in\bfS_{aux}$ and $i\in\{1, \ldots, arity(R)\}$. The graph has two types of edges: 
\begin{enumerate}
    \item there is a ``normal'' edge from $(R,i)$ to $(S,j)$ if there is a variable
that occurs in position $i$ of an $R$-atom in the rule body
and that occurs in position $j$ of an $S$-atom in the rule head.
\item there is a ``special'' edge from $(R,i)$ to $(S,j)$ if
$R$ occurs in the rule body and there is an occurrence of $S$
in the rule head that has an existential variable in the 
$j$-th position.
\end{enumerate}
An $\edatalog$ program is said to be \emph{weakly acyclic} if its dependency graph does not contain a directed cycle going through a special edge.

\paragr{(Unions of) Conjunctive Queries}
For $\bfS$ a schema and $k\geq 0$, 
a \emph{$k$-ary conjunctive query (CQ) over $\bfS$} is 
an expression of the form
\begin{equation}\label{eq:cq}
    \phi(y_1, \ldots, y_k) \colondash \exists\textbf{x}(\phi_1\land\cdots\land\phi_n)
\end{equation} 
where each $\phi_i$ is a relational atomic formula, 
and such that each variable $y_i$ occurs in at least one 
conjunct $\phi_j$. A \emph{$k$-ary union of conjunctive queries (UCQ) over $\bfS$} is a finite disjunction of
$k$-ary CQs over $\bfS$. We denote by $q(I)$ the set
of tuples $\textbf{a}$ for which it holds
that $I\models q(\textbf{x})$. 

The \emph{canonical instance} of a CQ of the form 
(\ref{eq:cq}) is the pointed instance $(I,\textbf{y})$
where $I$ is the instance with active domain
$\{y_1, \ldots, y_n, \textbf{x}\}$ whose facts are
the conjuncts of $\phi$, and $\textbf{y}=y_1 \ldots y_k$.
Conversely, the \emph{canonical CQ} of a pointed
instance $(I,\textbf{a})$ with $\textbf{a}=a_1 \ldots a_k$,
is obtained by (i) associating a unique variable $y_a$ to 
each $a\in adom(I)$, (ii) letting $\textbf{x}$ be an
enumeration of all variables $x_a$ for $a\in adom(I)\setminus\{a_1, \ldots, a_n\}$, and (iii) taking
the query $q(y_{a_1}, \ldots, y_{a_n}) \colondash \exists \textbf{x}\bigwedge_{R(b_1, \ldots, b_n)\in I}R(y_{b_1}, \ldots, y_{b_n})$. The well known Chandra-Merlin
theorem states that a tuple $\textbf{a}$ belongs to
$q(I)$ if and only if the canonical instance of $q$
homomorphically maps to $(I,\textbf{a})$.

\section{Proofs for Section~\ref{sec:tam}}

\propalmostmonadic*

\begin{proof}
Suppose, for the sake of a contradiction, that 
there was an equivalent monadic Datalog program $P$.
Let $n$ be the maximum number of variables in any
$Ans$ rule of $P$.
Consider the $\bfS_{in}$-instance $I$
consisting of the facts $R(a_0,a_1), R(a_1, a_2), \ldots, R(a_n, a_{n+1})$
as well as the facts $R(b_0,b_1), R(b_1, b_2), \ldots, R(b_n, b_{n+1})$.
Then $Ans(a_0, a_{n+1})$ is a fact of $P(I)$ while
$Ans(a_0, b_{n+1})$ is not.
A simple isomorphism argument shows that, for all
$i\leq n+1$ and for all $S\in \bfS_{aux}$, 
$S(a_i)$ belongs to $\chase_P(I)$ if and only if $S(b_i)$ belongs to $\chase_P(I)$.
It is then easy to see that any derivation of $Ans(a_0, a_{n+1})$
using a rule of $P$ implies also the existence of a derivation 
of $Ans(a_0, b_{n+1})$ using the same rule. A contradiction. 
\end{proof}

\thmmonadicreduction*

\begin{proof}
To simplify the exposition,
we may assume that the articulation position of each relation (if it has one) is the
first position.
Let $\mathcal{Q}=\{Q_1, \ldots, Q_k\}$.
For each relation $S\in \bfS_{out}^P\cup\bfS_{aux}^P$
with $arity(S)>0$,
and
for each partial function $f:\{1, \ldots, arity(S)\}\hookrightarrow \mathcal{Q}$,
we create a unary relation $S^f$. 
The intuitive meaning of 
$S^f(x)$ is: $$\exists y_1\ldots y_k (S(y_1, \ldots, y_k)\land x=y_1\land \bigwedge_{f(i)=Q_j}Q_j(y_i))~.$$
Let $\bfS'_{aux}$ be the set of all these new unary relations.
Finally, we define the set $\Sigma^{P'}$ of rules  of our new program $P'$.
Take any rule in $\rho\in\Sigma^P$. Without loss of generality, we can 
we can assume that $\rho$ is of the form
\[ R_0(\textbf{x}_0) \colondash R_1(\textbf{x}_1), \ldots, R_n(\textbf{x}_n), E_{n+1}(\textbf{x}_{n+1}), \ldots, E_{n+m}(\textbf{x}_{n+m}) \]
where each $R_i\in\bfS_{aux}^P\cup\bfS_{out}^P$ 
and each $E_i\in \bfS_{in}$.
For $0\leq i\leq n$, let
$f_i:\{1, \ldots, arity(R_i)\}\hookrightarrow \mathcal{Q}$ be a partial function, such that
the following consistency requirement is
satisfied: whenever a variable occurs in 
multiple $\bfS_{aux}\cup\bfS_{out}$-atoms in the above rule, 
say,
in the $j$-th argument position of the atom
$R_i(\textbf{x}_i)$ and in the $j'$-th 
argument position of the atom $R_{i'}(\textbf{x}_{i'})$, then
 $f_i(j)=f_{i'}(j')$ (we allow here that $f_i(j)$ and $f_{i'}(j')$ are both undefined).

For each  rule $\rho\in \Sigma^P$ and for each choice of
partial functions $f_0, \ldots, f_n$, 
satisfying the above consistency requirement, 
we add to $\Sigma^{P'}$ the rule

\[ R_0^{f_0}(x_{0,1}) \colondash R^{f_1}_1(x_{1,1}), \ldots, R_n^{f_n}(x_{n,1}), 
E_{n+1}(\textbf{x}_{n+1}), \ldots, E_{n+m}(\textbf{x}_{n+m}),
\bigwedge_{f_0(i)=Q_j} Q_{j}(x_{0,i})
\]
where $x_{i,j}$ stands for the $j$-th variable in the tuple of variables $\textbf{x}_i$.

Finally we add the rule
\[ Ans() \colondash R^f(x)\]
where $R$ is the relation mentioned in the statement of the proposition,
and $f:\{1, \ldots, arity(R)\}\to \mathcal{Q}$ is the total function given by $f(i)=Q_i$.
This concludes the definition of the monadic Datalog program $P'$. 

Let $I$ be any $\bfS_{in}^P$-instance, and let $I'=I\cup\{Q_1(a_1), \ldots, Q_k(a_k)\}$.
\begin{claim*}
                The following are equivalent, for all 
                $R^f\in \bfS'_{aux}$ and $c\in adom(I)$:
                \begin{enumerate}
                    \item                 $c\in (R^f)^{P'(I)}$ 
                    \item there is a tuple $(b_1, \ldots, b_n)\in R^{P(I)}$
                           such that $b_1=c$ and, for all $i\leq n$, if $f(i)=Q_j$, then
                           $b_i=a_i$.
                \end{enumerate}
\end{claim*}
Both directions of this claim can be proved by an induction on the length of derivations.
In 
particular, it follows from this claim that $Ans() \in P'(I')$
iff $R(a_1, \ldots, a_k)\in P(I)$.

Now let us prove the converse direction. For every $S\in\bfS^{P'}_{aux}\cup\{Q_1,\dots,Q_k\}$, 
$\bfS_{aux}^P$ contains a $(1+k)$-ary symbol $S^*$. We shall use 
${Ans}^*$ (where $Ans$ is the output predicate in $P'$) to denote the output predicate $R$ of $P$ 
The intuitive meaning of $S^*(\textbf{x},y_1,\dots,y_k)$ is

$$S(\textbf{x})\land \bigwedge Q_j(y_j)$$

We note that $\textbf{x}$ consists of a single variable whenever $S\in\bfS^{P'}_{aux}\cup\{Q_1,\dots,Q_k\}$
and is empty whenever $S=Ans$.

To achieve this we include to $\Sigma^P$ the following rules. First, for every $Q_i$ we add the rule:
$$Q^*_i(x_i,x_1,\dots,x_k) \colondash \text{(empty body)}$$
We note that  although this rule is unsafe (that is, the variables in the head do not occur in the body) this can be easily fixed
extending the rule body with an $\bfS_{in}$-atom 
containing every variable in the head and with fresh variables in all other positions of the atom (there are multiple ways to do this, and we add all safe rules that can be obtained in this way).

Secondly, for each rule $\rho$ in $\Sigma^{P'}$ we add to $\Sigma^{P}$ a new rule obtained from $\rho$. Without loss of generality, we can 
we can assume that $\rho$ is of the form
\[ R_0(\textbf{x}_0) \colondash R_1(x_1), \ldots, R_n(x_n), E_{n+1}(\textbf{x}_{n+1}), \ldots, E_{n+m}(\textbf{x}_{n+m}) \]
where $R_0\in\bfS_{aux}^{P'}\cup \bfS_{out}^{P'}$, each $R_i\in\bfS_{aux}^{P'}\cup\{Q_1,\dots,Q_k\}$, and 
each $E_i\in \bfS_{in}^{P'}\setminus\{Q_1,\dots,Q_k\}$. Then we add to $\Sigma^P$ the rule:
\[ R_0(\textbf{x}_0,y_1,\dots,y_k) \colondash R_1(x_1,y_1,\dots,y_k), \ldots, R_n(x_n,y_1,\dots,y_k), E_{n+1}(\textbf{x}_{n+1}), \ldots, E_{n+m}(\textbf{x}_{n+m}) \]

Let $I$ be any $\bfS_{in}^P$-instance, and let $I'=I\cup\{Q_1(a_1), \ldots, Q_k(a_k)\}$. The following claim can be
proved by induction on the length of the derivations.
\begin{claim*}
                The following are equivalent, for all 
                $S^*\in \bfS^{P}_{aux}$ and $c\in adom(I)$:
                \begin{enumerate}
                    \item $c\in S^{P'(I)}$ 
                    \item $(c,a_1,\dots,a_k)\in ({S^*})^{P(I)}$
                \end{enumerate}
\end{claim*}
Note that it follows that $Ans() \in P'(I')$
iff $Ans^*(a_1, \ldots, a_k)\in P(I)$.
\end{proof}

\coralmostmonadictomso*
\begin{proof}
Let $P'$ be as in Theorem~\ref{thm:monadic-reduction}.
By Theorem~\ref{thm:monadic-mso}, there is an MSO sentence
$\phi$ such that, for all $\bfS\cup\{Q_1, \ldots, Q_k\}$-instances $I$,
$Ans() \in P'(I)$ iff
$I\models\phi$.
Let $$\psi(x_1, \ldots, x_k)= \exists Q_1\ldots Q_k (\phi\land \bigwedge_{i}\forall z(Q_i(z)\leftrightarrow z=x_i))$$
Then, for all $\bfS_{in}^P$-instances $I$, $I\models\psi(a_1, \ldots, a_k)$ iff
$R(a_1, \ldots, a_k)\in P(I)$.
\end{proof}

\lemtreeunfoldings*

\begin{proof}
Recall that $\Unfoldings(P,R)$ consists of canonical 
instances of derivable rules, where a derivable rule
is a rule can be obtained from the rules of $P$ through
the operation of substituting occurrences of a rule
head by the corresponding rule body. It is easy to see
that this substitution operation preserves 
tree-shapedness, and therefore every derived rule is
tree-shaped. It follows that every $\Unfoldings(P,R)$
consists of acyclic pointed instances.
\end{proof}

\thmmsototam*

\smallskip

The proof is deferred to Appendix~\ref{app:msototam}.

\corclosureundercomposition*

\begin{proof}
The composition of $P_1$ and $P_2$ is clearly
expressible as a tree-shaped Datalog program:
we may assume that $\bfS_{aux}^{P_1}$ and $\bfS_{aux}^{P_2}$ are disjoint. 
Let $P_{3} = (\bfS_{in}^{P_1},\bfS_{out}^{P_2},\bfS_{aux}^{P_1}\cup\bfS_{out}^{P_1}\cup\bfS_{aux}^{P_2},\Sigma^{P_1}\cup\Sigma^{P_2})$. Then $P_{3}$ defines the composition of $P_1$ and $P_2$. Note that $P_{3}$ is tree-shaped but no longer necessarily almost-monadic. As we will show, however, $P_{3}$ is nevertheless equivalent to a
TAM Datalog program.

For each $k$-ary relation $R\in \bfS_{out}^{P_1}$, let 
$\phi_R(x_1, \ldots, x_k)$ be the MSO query over schema $\bfS_{in}^{P_1}$ defined by $(P_1,R)$.
Similarly, for each $k$-ary relation $S\in \bfS_{out}^{P_2}$, let 
$\phi_S(x_1, \ldots, x_k)$ be the MSO query over schema $\bfS_{in}^{P_2}$ defined by $(P_2,S)$.
We can substitute, in $\phi_S$, all occurrences of relation symbols $R\in\bfS_{out}^{P_1}$
by their defining formula $\phi_R$. In this way, 
we obtain, for each $S\in\bfS_{out}^{P_2}$, 
an MSO query $\phi'_S$ over the schema $\bfS_{in}^{P_1}$.
Note that $\phi'_S$ is precisely the MSO query defined by 
$(P_{3},S)$.

It follows by Theorem~\ref{thm:mso-to-tam} that, for each
$S\in\bfS_{out}^{P_2}$, the MSO query $\phi'_S$ is 
definable by a TAM Datalog program. As a last step, we
merge the TAM Datalog programs in question to 
obtain a single TAM Datalog program that is equivalent to $P_{3}$.
\end{proof}

\thmnormalform*

\begin{proof}
First, we will show how to ensure that each rule contains
at most one occurrence of a relation from $\bfS_{in}$. 
Consider any rule whose body has two or more conjuncts 
involving relations from $\bfS_{in}$. Since the
program is tree-shaped, the incidence graph of the rule
body is acyclic. It follows that the rule in question can be
written (by re-ordering the atoms in the body as needed) as
follows:
\[ R_0(\textbf{x}_0) \colondash R_1(\textbf{x}_1), \ldots,
 R_i(\textbf{x}_i), R_{i+1}(\textbf{x}_{i+1}), \ldots,
R_n(\textbf{x}_n) \]
where one of the relations $R_1, \ldots, R_i$ is in 
$\bfS_{in}$,  one of the relations $R_{i+1}, \ldots, R_n$
is in $\bfS_{in}$, and  the intersection
$\{\textbf{x}_1, \ldots, \textbf{x}_i\}\cap\{\textbf{x}_{i+1}, \ldots, \textbf{x}_n\}$ contains at most one variable $z$.
Indeed, if the program is connected, such a variable $z$ must exist.

Let $\textbf{u}$ be an enumeration of the variables in 
$\{\textbf{x}_{i+1}, \ldots, \textbf{x}_n\}$ without duplicates, and starting with $z$, or otherwise starting with any variable occurring
in an input-relation atom in $R_{i+1}(\textbf{x}_{i+1}), \ldots,
R_n(\textbf{x}_n)$.

We can replace the above rule by the following two rules:
\begin{align*}
R_0(\textbf{x}_0) &\colondash R_1(\textbf{x}_1), \ldots,
 R_i(\textbf{x}_i), R'(\textbf{u}) \\
 R'(\textbf{u}) &\colondash R_{i+1}(\textbf{x}_{i+1}), \ldots,
R_n(\textbf{x}_n)
\end{align*}
where $R$ is a fresh auxiliary relation of suitable arity, whose
articulation position is the first position.
Observe that both new rules have strictly fewer occurrences of 
relations from $\bfS_{in}$ than the original rule, and
that this construction preserves connectedness. 
If we repeat this process, we will end up with at most a linear
number of rules, each of size no greater than the size of the
original rule. Furthermore, this can clearly be performed in 
polynomial time.

Next, we explain how to ensure that each rule body contains
at least one (hence, exactly one) 
relation from $\bfS_{in}$. Here we use
the fact that every tuple that
is derived into a defined relation must consist of values
originating from facts of the input instance. Specifically, given a
rule
\[ R_0(\textbf{x}_0) \colondash R_1(\textbf{x}_1), \ldots,
R_n(\textbf{x}_n) \]
not containing any relation from $\bfS_{in}$. 
Let $x$ be any variable occurring in the articulation position of 
one of the atoms in the rule body, and replace the rule by all 
possible rules that extend its body with an additional atom
$R'(\textbf{u},x,\textbf{v})$ with $R'\in\bfS_{in}$ and $\textbf{u},\textbf{v}$ distinct, fresh variables.
Again, 
connectedness is preserved.
\end{proof}

\section{Proofs for Section~\ref{sec:adjoints}}

\thmcomposition*

\begin{proof}
It suffices to define $\Omega_{P_2\cdot P_1}(J) = \{(J'',\kappa\cdot\iota)\mid
(J',\iota)\in \Omega_{P_2}(J), (J'',\kappa)\in \Omega_{P_1}(J')\}$, cf.~the following commuting diagram:

  \[ \begin{tikzcd}[baseline=(I.base)]
P_2(P_1(I)) \arrow{r}{} & J  \\%
P_1(I) \arrow[hook]{u}{id} \arrow{r}{} & J'\arrow[swap,hook]{u}{\iota} & (J',\iota)\in \Omega_{P_2}(J) \\%
|[alias=I]|I \arrow{r}{} \arrow[hook]{u}{id}    & J''\arrow[swap,hook]{u}{\kappa} &
(J'',\kappa)\in \Omega_{P_1}(J')
\end{tikzcd} \qedhere
\]
\end{proof}

\thmtamadjoint*

\begin{proof}
We first consider the special case of connected programs.
We may assume without loss of generality that $P$ is simple. 
This means that every rule is of the following form:
\begin{align}
R_0(\textbf{x}_0) &\colondash E(\textbf{y}), R_1(\textbf{x}_1), \ldots, R_m(\textbf{x}_m) \\
R(\textbf{x}) &\colondash E(\textbf{y}), R_1(\textbf{x}_1), \ldots, R_m(\textbf{x}_m)
\end{align}
where $E$ is an input relation, each $R_i$ is an auxiliary relation, and $R$ is an output relation.

To simplify the exposition below, we introduce some further
notation. For each atom $R_i(\textbf{x}_i)$ as in the above
rule types, we will denote by $p_i\in\{1, \ldots, n\}$ (with $n=arity(E)$) the 
 unique number such that $y_{p_i}$ is equal to the articulated variable
in $R_i(\textbf{x}_i)$. It indeed follows from the definition of TAM Datalog and the assumed connectedness and simplicity of $P$ that such an index exists and is unique.

We construct an $\bfS_{in}$-instance $J'$.
The instance $J'$ consists of all facts $E((b_1,X_1), \ldots, (b_n,X_n))$ where
\begin{enumerate}
    \item Each $b_i$ is an element of $adom(J)\cup\{\bot\}$ and $X_i$ is a set of $\bfS_{aux}$-facts 
    over $adom(J)\cup\{\bot\}$ (not necessarily facts of $J$) in which $b_i$ occurs in articulation position,
    \item For each rule of the form (1) above and for each map $g:\{\textbf{y},\textbf{x}_1, \ldots, \textbf{x}_m\}\to adom(J)\cup\{\bot\}$, if for each $1\leq i\leq m$, $R_i(g(\textbf{x}_i))\in X_{p_i}$ 
    then $R_0(g(\textbf{x}_0))\in X_{p_0}$. 
    \item For each rule of the form (2) above and for each map $g:\{\textbf{y},\textbf{x}_1, \ldots, \textbf{x}_m\}\to adom(J)\cup\{\bot\}$, if for each $1\leq i\leq m$, $R_i(g(\textbf{x}_i))\in X_{p_i}$ 
    then $R(g(\textbf{x}))$ is a fact of $J$. 
\end{enumerate}
Note that the total number of possible facts $E((b_1,X_1), \ldots, (b_n,X_n))$ in $J'$ is double exponential on the combined size of $P$ and $J$ and only exponential on $J$ if the arity of $P$ is bounded. 

Finally, $\Omega_P(J)=\{(J',\iota)\}$, where $\iota$ is the 
natural projection from $J'$ to $J$,
mapping all elements of the form 
$(a,X)$ to $a$ (and undefined on elements
of the form 
$(\bot,X)$).
\medskip

\begin{claim*}
For all $\bfS_{in}$-instances $I$, $P(I) \to J$ iff $I \to J'$. Moreover, the witnessing homomorphisms can be constructed so that the diagram in the statement of the theorem commutes.
\end{claim*}
\begin{claimproof}
{[$\Rightarrow$]} Let $h:P(I)\to J$. Recall that we denote by $\chase_P(I)$ the $\bfS_{in}\cup\bfS_{out}\cup\bfS_{aux}$-instance
that is the chase of $I$ (and of which $I$ and $P(I)$ are the
$\bfS_{in}$-reduct and 
$\bfS_{out}$-reduct, respectively).
We extend $h$ to the entire active domain of $\chase_P(I)$ 
by sending every element $a$ that is not in $adom(P(I))$ to $\bot$. 
With a slight abuse of notation, in what follows, we denote by $h$ the extended map from
$adom(\chase_P(I))$ to $adom(J)\cup\{\bot\}$.
For each $a\in adom(\chase_P(I))$, let 
$F_{a}$ be the set of all $\bfS_{aux}$-facts of $\chase_P(I)$ in which $a$ occurs in articulation position. We define $h'(a)=(h(a),h(F_{a}))$. 
We claim that $h'$ is a 
homomorphism from $I$ to $J'$. 
To prove this, let $E(a_1, \ldots, a_n)$ be any fact of $I$. 
We must show that the fact $E((h(a_1),h(F_{a_1})), \ldots, (h(a_n),h(F_{a_n})))$
belongs to $J'$. That is, we must show that the above conditions 1--3 are satisfied. 

Clearly, the first requirement is 
satisfied, namely, $h(X_{a_i})$ consists of facts in which $h(a_i)$
occurs in articulation position. 

To see that the second requirement holds, 
consider a rule of form (1) and any map $g:\{\textbf{y},\textbf{x}_1, \ldots, \textbf{x}_m\}\to adom(J)$, such that,
for each $1\leq i\leq m$, $R_i(g(\textbf{x}_i))\in X_{p_i}$.
By construction, this means that each fact $R_i(g(\textbf{x}_i))$ is the $h$-image of a fact $R_i(\textbf{b}_i)$ in $\chase_P(I)$. Now consider the map $\tilde{g}:\{\textbf{y},\textbf{x}_1, \ldots, \textbf{x}_m\}\to adom(I)$ defined 
by $\tilde{g}(\textbf{y})=(a_1\dots,a_n)$ and $\tilde{g}(\textbf{x}_i)=\textbf{b}_i$ for $i\leq i\leq m$. Note that $\tilde{g}$ is a well-defined function. Indeed, 
for every $i\leq i\leq m$ and every $z\in R(\bx_i)$, 
if $z$ occurs more than once in the rule, then $z$ must necessarily be the variable that appears in the articulation position $p_i$ of $R(\bx_i)$. It follows that $z$ occurs in $\by$ at position $p_i$, and, hence, necessarily, $R_i(g(\textbf{x}_i))$ belongs to the $h$-image of $X_{a_{p_i}}$, and, hence, $\tilde{g}(z)=a_{p_i}$.

Since $\chase_P(I)$ is closed under the rules of $P$, we may
conclude that $R_0(\tilde{g}(\bx_0))$ belongs to $\chase_P(I)$.

Let $z$ be the variable occurring in articulation position in 
$R_0(\textbf{x}_0)$. Recall that $z$ occurs in $\textbf{y}$ at
position $p_0$, and $\tilde{g}(z) = a_{p_0}$.
Then $R_0(\tilde{g}(\bx_0)) \in F_{a_{p_0}}$, and hence, 
$h(F_{a_{p_0}})$ contains $R_0(h\circ \tilde{g}(\bx_0))$. 
Note that by definition, $h\circ \tilde{g} = g$.
In particular, $R_0(h\circ \tilde{g}(\bx_0))=R_0(g(\bx_0))$. 
Therefore we have that
$R_0(g(\bx_0))\in h(F_{a_{p_0}})$, and we are done.

To see that the third requirement holds, 
consider a rule of form (2) and any map $g:\{\textbf{y},\textbf{x}_1, \ldots, \textbf{x}_m\}\to adom(J)$, such that,
for each $1 \leq i\leq m$, $R_i(g(\textbf{x}_i))\in X_{p_i}$.
By exactly the same reasoning as before,
an $h$-preimage of the rule head $R(g(\textbf{x}))$ belongs to $\chase_P(I)$. Hence, it belongs to $P(I)$,
therefore, $R(g(\textbf{x}))$ is a fact of $J$.

It is also clear from the construction that
$h \circ id = \iota \circ h'$, where $id$ is 
the identity function on $adom(I)\cap adom(P(I))$. That is, the 
diagram commutes.

{[$\Leftarrow$]} Conversely, let $h:I\to J'$. 
Note that $adom(\chase_P(I))=adom(I)$, and hence we $h(a)$ is well-defined
for all $a\in adom(\chase_P(I))$.
Let $h':adom(I)\to adom(J)$ be the map such that $h'(a)=b$ whenever $h(a)=(b,X)$.

\begin{subclaim}
For all $a\in adom(\chase_P(I))$, if $h(a)=(b,X)$, then the $h'$-image of every $\bfS_{aux}$-fact of $\chase_P(I)$ in which $a$ occurs in articulation position belongs to $X$.
\end{subclaim}

\begin{subclaim}
 $h'$ is a homomorphism from 
$P(I)$ to $J$. 
\end{subclaim}

Subclaim 1 can be proved by induction on the derivation length of the 
fact in question. 

To prove subclaim 2, let 
$R(\textbf{a})$ be an $\bfS_{out}$-fact belonging to $P(I)$.
Its derivation must use a rule of the form (2) above,
using an assignment $g$ (where $g(\textbf{x})=\textbf{a})$.
By Subclaim 1, we have that that $R_i(h'(g(\textbf{x}_i)))$ 
belongs to $h(g(y_{p_i}))_2$, for $y_{p_i}$  the
 articulated variable in $\textbf{x}_i$.
 Furthermore, $E(g(\textbf{y}))$ holds in $I$, and hence
 $E(h(g(\textbf{y})))$ holds in $J'$.
 By construction of 
 $J'$, this means that the $R(h'(g(\textbf{x})))$,
 that is, $R(h'(\textbf{a}))$, belongs
 to $J$. This concludes the proof for the case of \emph{connected} TAM Datalog programs.

It is also clear from the construction that
$\iota \circ h = h' \circ id$, where $id$ is 
the identity function on $adom(I)\cap adom(P(I))$. That is, the 
diagram commutes.
\end{claimproof}

 Finally, we show how to handle non-connected TAM Datalog programs. 
 Let $P$ be a non-connected TAM Datalog program. Let $P'$ be obtained
 from $P$ by adding a fresh binary input-relation $S$, and using 
 this relation to make every every rule connected in some arbitrary way
 (more precisely, whenever the incidence graph of a rule body 
 has multiple connected component, we add $S$-atoms to the body connecting
 these components while preserving tree-shapedness and almost-monadicity.
 For every input instance $I$, we denote by $\widehat{I}$ the 
 $\bfS_{in}\cup\{S\}$-instance extending $I$ with all facts
 of the form $S(a,b)$ for $a,b\in adom(I)$. 
 Furthermore, given an instance $J'$ over the schema $\bfS_{in}\cup\{S\}$, 
 by an ``$S$-component'' of $J'$ we will mean the $\bfS_{in}$-retract of a fully $S$-connected sub-instance of $J'$. Clearly, if $J$ is an $\bfS_{in}$-instance and $J'$ is a $\bfS_{in}\cup\{S\}$-instance, then
 $\widehat{J}\to J'$ iff $J\to J''$ for some $S$-component $J''$ of $J'$.
 Now we simply define $\Omega_P(J)$ to be the set of all $S$-components of
 instance in $\Omega_{P'}(J)$.
Then we have:
 $P(I)\to J$ iff $P'(\widehat{I})\to J$ iff
 $\widehat{I}\to \Omega_{P'}(J)$ iff
 $I\to J'$ for some $J'\in \Omega_P(J)$.

 As a side remark, we mention that there
 is another way present the final argument
 where we lift the connected
 case to the general case: we can 
 view the function that sends $I$ to $\widehat{I}$ as a functor that itself
 has a generalized right-adjoint (sending $I$ to its $S$-connected components). Thus, 
 we can argue by composition, using 
 Theorem~\ref{thm:composition}. 
\end{proof}

\proptreeshapednoadjoint*

\begin{proof}
Assume towards a contradiction that $P$ has a generalized right-adjoint $\Omega_P$. 
Let $J$ be
the two-element $\{R\}$-instance consisting of the facts $R(0,1)$ and $R(1,0)$.

For $n\geq 1$, let 
 $C_n$ be the $\{E,F\}$-instance consisting of the facts
 $E(v_0, v_1)$, \dots, $E(v_{n-1}, v_n)$, $E(v_n,v_0)$, that is, 
 the directed $E$-cycle of length $n$.
Trivially, $P(C_n)\to J$ for all $n\geq 1$. Therefore, for 
each $n\geq 1$, we have 
$C_n\to J'$ for some $J'\in\Omega_P(J)$. For every $J'\in\Omega_P(J)$ and for each element $b$ of $J'$, let us define $n_{J',b}$ to be an arbitrarily chosen 
value such that $(C_n,v_0)\to (J',b)$, or undefined,
if no such value exists. It follows from our earlier observation
that $n_{J',b}$ is defined for at least one pair $(J',b)$ with
$J'\in\Omega_P(J)$.
Let $m$ be a common multiple of all defined $n_{J',b}$’s.

For every pair of positive integers $e \leq f$, let $I_{e,f}$ be the $\{E, F\}$-instance depicted as follows:
\[
u_0 \xrightarrow{E} v_0 \xrightarrow{F} u_1 \xrightarrow{E} v_1 \xrightarrow{F} u_2
\xrightarrow{\text{\begin{tabular}{l}sequence of\\[-1mm]$e$ $E$-edges\end{tabular}}} z
\xrightarrow{\text{\begin{tabular}{l}sequence of \\[-1mm] $f$ $F$-edges\end{tabular}}} u_0
\]

\begin{myclaim}
For all $1\leq e\leq f$, the following are equivalent:
\begin{enumerate}
    \item $I_{e,f}\to J'$ for some $J'\in\Omega_P(J)$
    \item $e\neq f$.
\end{enumerate} 
\end{myclaim}
\begin{claimproof}
If $e = f$ then $P(I_{e,f})$ contains an
$R$-cycle of odd length, viz.~$u_0 \xrightarrow{R} u_1\xrightarrow{R} u_2\xrightarrow{R} u_0$, and therefore $P(I_{e,f})\not\to J$. Hence, $I_{e,f}\not\to J'$ for all $J'\in \Omega_P(J)$.
On the other hand, if
$e < f$, then $P(I_{e,f})$ is a disjoint union of $R$-paths, and, clearly, $P(I_{e,f})\to J$. Therefore, $I_{e,f}\to J'$ for some
$J'\in\Omega_P(J)$.
\end{claimproof}

Now, let $e$ be larger than the universe of all instances in $\Omega_P(J)$ and let $f = e + m$. By Claim 1, there is a homomorphism
 $h:I_{e,f}\to J'$ for some $J'\in\Omega_P(J)$. 
 We will show that $h$ can be extended to a homomorphism $h':
I_{f,f}\to J'$, which contradicts Claim 1. 

Let \[u_2 = x_0 \xrightarrow{E} x_1 \cdots \xrightarrow{E} x_e = z\] be the sub-instance of $I_{e,f}$ consisting
of the $E$-edges joining $u_2$ and $z$.
Similarly, let
\[u_2 = x_0 \xrightarrow{E} x_1 \cdots \xrightarrow{E} x_e \xrightarrow{E} x_{e+1} \ldots \xrightarrow{E} x_f = z\]
be the sub-instance of $I_{f,f}$ consisting
of the $E$-edges joining $u_2$ and $z$. Recall that $f=e+m$.
Since $e$ is larger
than the domain size of $J'$, it must be the case that
$h(x_i)=h(x_j)=b$ for some $i<j\leq e$, for some element $b$ of
$J'$. This means that $b$ lies on a directed $E$-cycle in $J'$, and hence, in particular, it lies on a directed $E$-cycle of length $m$, say, $b=b_0 \xrightarrow{E} b_1 \cdots \xrightarrow{E} b_m = b$.
The mapping $h':I_{f,f}\to J'$ can be constructed simply by extending $h$ and mapping $x_{e+i}$ to $b_i$ for $1\leq i\leq m$. 
\end{proof}

\propcacyclicnoadjoint*

\begin{proof}
Let $P$ be the Boolean Datalog program in question.
Note that $\Unfoldings(P,Ans)$ consists of a pointed structure
that is c-acyclic but not acyclic.
It does not admit a generalized right-adjoint: let $J$ be the empty
$\bfS_{out}^P$-instance, and suppose for the 
sake of contradiction that there is a finite
set $\{J_1, \ldots, J_n\}$ such that,
for all $\bfS_{in}^P$-instances $I$, 
$P(I)\to J$ iff $I\to J_i$ for some $i\leq n$. Let $I_c$ be the
instance consisting of a single reflexive $Ans$-edge of the form $Ans(a,a)$.
Clearly, $P(I_c)\not\to J$, and therefore,
$I_c\not\to J_i$. That is, $J_1, \ldots, J_n$ do not contain a reflexive $Ans$-edge. Next, let $I_n$ be the instance that is an (irreflexive) $Ans$-clique of size $n$, where $n$ is an arbitrary 
number greater than the size of each $J_i$. Then, $I_n\not\to J_i$ (because it there was such a homomorphism, $J_i$ would necessarily contain a reflexive $Ans$-edge), but, trivially, $P(I_n)\to J$.
\end{proof}

The next lemma lists the main differentiating properties of strongly linear \edatalog programs that we will make use of.

\begin{lemma}
\label{lem:linear-properties}
Let $P=( \bfS_{out}, \bfS_{aux}, \Sigma)$ be a strongly linear
\edatalog program and let $I,I'$ be $\bfS_{in}$-instances.
\begin{enumerate}[(a)]
\item $P(I\cup I')\leftrightarrow_{adom(I\cup I')} P(I)\cup P(I')$.
\item For every solution $J$ of $I$, and for every function $f$ with $dom(f)=adom(I)$,
$f[J]$ is a solution for
$f[I]$, where $f[K]$ denotes the instance obtained
from $K$ by replacing every value $a\in dom(f)$ by $f(a)$. 
\end{enumerate}
\end{lemma}

\begin{proof}
\begin{enumerate}[(a)]
\item Recall that $P(I)$ can be defined as the $\bfS_{out}$-reduct of an
(arbitrarily chosen) universal solution for $I$ with respect to $P$. 
Let $J, J', J''$ be universal solutions for $I$, $I'$, and $I\cup I'$,
respectively. We may assume without loss of generality that 
$adom(J)\cap adom(J')\subseteq adom(I\cup I')$. 
It is easy to see that $J\cup J'$ is a solution for $I\cup I'$.
(Indeed, whenever the body of a strongly linear rule is satisfied in $J\cup J'$, then it is
satisfied in $J$ or in $J'$, and hence, the rule head is also satisfied
in the same instance, therefore also in $J\cup J'$.) Therefore, 
by definition of universal solutions, $J''\to_{adom(I\cup I')} J\cup J'$.
Conversely, since $J''$ is a solution for both $I$ and $I'$, we have,
by the definition of universal solutions, that 
$J\to_{adom(I\cup I')} J''$ and $J'\to_{adom(I\cup I')} J''$. It follows
that $J\cup J'\to_{adom(I\cup I')} J''$. In conclusion,
$J\cup J' \leftrightarrow_{adom(I\cup I')} J''$. Therefore, the same 
relationship holds between their $\bfS_{out}$-reducts.

\item
Since $I\subseteq J$, clearly, also $h[I]\subseteq h[J]$.
Furthermore, $h[J]$ is closed under the rules of the program:
if the body a strongly-linear constraint is satisfied by some tuple $\textbf{a}$
in $h[J]$, then some tuple $\textbf{b}\in h^{-1}[\textbf{a}]$  satisfies the same
rule body in $J$. Therefore, the head of the rule is satisfied in $J$
for $\textbf{b}$, and hence in $I$ for $\textbf{a}$.
\qedhere
\end{enumerate}
\end{proof}

\thmlinearadjoints*

\begin{proof}
Let $P=(\bfS_{in}, \bfS_{out}, \bfS_{aux}, \Sigma)$, and let $J$
be any $\bfS_{out}$-instance. 
Let $D=adom(J)\cup\{\bot\}$ where $\bot$ is a fresh value.
We define $J'$ to be the $\bfS_{in}$-instance
 consisting of all facts $R(\textbf{d})$ over domain $D$ for which
it holds that $P(\{R(\textbf{d})\})\to_{adom(J)} J$.

It follows from~\cite[Corollary 5.10]{Benedikt2021:inference} that $J'$ can be computed
in ExpTime, and in PTime if $P$ is fixed. For 
completeness, we sketch a direct argument here:
it is easy to see that $J'$ is precisely the $\bfS_{in}$-reduct of the maximal $\bfS$-instance $K$ over the domain $D$ satisfying the following two conditions:
\begin{enumerate}
  \item The $\bfS_{out}$-reduct of $K$ is contained in $J$.
    \item All rules of $P$ are satisfied in $K$.
\end{enumerate}
We note that $K$ can be computed greedily in ExpTime by first setting it to the maximal instance satisfying (1) above and  then iteratively removing any 
$\bfS_{in}\cup\bfS_{aux}$-fact $S(\textbf{d})$ that violates any rule
\[ \exists \textbf{z} \big(R_1(\textbf{x}_1), \ldots, R_n(\textbf{x}_n)\big) \colondash S(\textbf{y})
\]
of $P$.
Further, note that if the arity of the $\bfS_{in}$-relations is bounded then $J'$ has polynomial size and that, in addition, $J'$ can be computed in polynomial time if $P$ is fixed.

It follows from the definition of $J'$ that $P(J')\to_{adom(J)} J$
(cf.~Lemma~\ref{lem:linear-properties}(a)).
Let $\iota:adom(J')\hookrightarrow adom(J)$ 
be the identity function on $adom(J)$.
We claim that $(J',\iota)$ serves as a right-adjoint.

In one direction, let $I$ be any $\bfS_{in}$-instance. 
If there is a homomorphism $h:I\to J'$, then
we obtain the following commuting diagram,
where $h'$ is given by Lemma~\ref{lem:edatalog-monotonicity}:
  \[ \begin{tikzcd}[column sep=large,row sep=large]
P(I) \arrow{r}{h'} & P(J') \arrow[pos=.9,swap,shorten >=20pt]{r}{adom(J)} & ~~~~~J~~~~~ \\%
I \arrow{r}{h} \arrow[hook]{u}{id}    & J'\arrow[swap,hook]{u}{id} \arrow[hook,bend right=20]{ur}{\iota} &
\end{tikzcd}
\]

Conversely, let 
$h:P(I)\to J$. Let $h'$ extend $h$ by mapping
all elements of $adom(I)\setminus adom(P(I))$
to $\bot$. We claim that $h':I\to J'$.
Consider any fact $S(\textbf{b})$ of $I$.
We must show that $S(h'(\textbf{b}))\in J'$, or, in 
other words, that
$P(\{S(h'(\textbf{b}))\})\to_{adom(J)} J$. 
Since $\{S(h'(\textbf{b}))\}\subseteq h'[I]$, 
every solution for the latter is a solution for
the former. By Lemma \ref{lem:linear-properties}(b), $h'[P(I)]$ is a solution
for $h'[I]$. Since $h$ and $h'$ agree on $adom(P(I))$,
$h'[P(I)]=h[P(I)]$. Putting everything
together, we have that 
$h[P(I)]$ is a solution for $\{S(h'(\textbf{b}))\}$.
By the definition of universal solutions,
there is a homomorphism 
$g:P(\{S(h'(\textbf{b}))\})\to_{\{h'(\textbf{b})\}} h[P(I)]\subseteq J$.
This concludes the proof that $h':I\to J'$.
Furthermore, it is clear from the construction that the following diagram commutes:
  \[ \begin{tikzcd}[baseline=(I.base)]
P(I) \arrow{r}{h} & J \\%
|[alias=I]| I \arrow{r}{h'} \arrow[hook]{u}{id}    & J'\arrow[hook,swap]{u}{\iota} &
\end{tikzcd} \qedhere
\]

\end{proof}

\section{Proofs for Section~\ref{sec:dualities}}

\thmcacyclicdualititestgd*

\begin{proof}
Let $D$ be the finite set such that $(\{(A',\textbf{a})\mid (A,\textbf{a})\in F\},D)$ is a homomorphism duality, 
given by Theorem~\ref{thm:cacyclic-duals}.
Let 
$D' = \{(P(B'),\textbf{b}')\mid (B,\textbf{b})\in D, 
(B',\iota)\in \Omega_{P_\Sigma}(B), 
\iota(\textbf{b}')=\textbf{b}\}$.
Note that $D'$ consists of pointed instances satisfying $\Sigma$. 
We will show that $(F,D')$ is a homomorphism duality w.r.t.~$\Sigma$.

Let $(C,\textbf{c})$ be a pointed instance with $C\models \Sigma$.
The following chain of equivalences holds:
\begin{align}
    (C,\textbf{c}) &\in F\upclosure \tag{1} \\
    &\Updownarrow \text{Lemma~\ref{lem:edatalog-monotonicity} and $P_\Sigma(C) \leftrightarrow_{adom(C)} C$} \nonumber\\
    (C,\textbf{c})&\in \{(A',\textbf{a})\mid (A,\textbf{a})\in F\}\upclosure \tag{2} \\
    &\Updownarrow \text{duality assumption} \nonumber\\
    (C,\textbf{c})&\not\in D\downclosure \tag{3} \\
    &\Updownarrow \text{to be proved below} \nonumber\\
    (C,\textbf{c})&\not\in D'\downclosure \tag{4}
\end{align}
It remains to prove the equivalence between (3) and (4). 

From (3) to (4): By contraposition: suppose 
that $(C,\textbf{c})\to (P_\Sigma(B'),\textbf{b}')$ for some $(B,\textbf{b})\in D, 
(B',\iota)\in \Omega_{P_\Sigma}(B)$, and $\iota(\textbf{b}')=\textbf{b}$.
Trivially, we have $id:(B',\textbf{b}')\to (B',\textbf{b}')$. 
It follows by the generalized adjoint 
property that $(P_\Sigma(B'),\textbf{b}')\to (B,\iota(\textbf{b}'))$.
Therefore, by transitivity, and since $\iota(\textbf{b'})=\textbf{b}$, 
we have $(C,\textbf{c})\to (B,\textbf{b})$ and therefore $(C,\textbf{c})\in D\downclosure$.

From (4) to (3): Again, by contraposition: assume
$(C,\textbf{c})\in D\upclosure$.
Since $P_\Sigma(C)\leftrightarrow_{adom(C)} C$, it follows that $(P_\Sigma(C),\textbf{c})\to 
(B,\textbf{b})$ for some $(B,\textbf{b})\in D$. 
It follows by the adjoint property that 
$(C,\textbf{c})\to (B',\textbf{b}')$ for some $(B,\textbf{b})\in D$, $(B',\iota)\in \Omega_{P_\Sigma}(B)$, and $\textbf{b}'\in \iota^{-1}(\textbf{b})$.
Then also
$(C,\textbf{c})\to (P_\Sigma(B'),\textbf{b}')$.
This means that 
$(C,\textbf{c})\in D'\downclosure$.
\end{proof}

By the \emph{c-girth} of a pointed instance $(A,\textbf{a})$ we will
mean the length of the smallest cycle in the incidence graph of $A$
that does not pass through any element in $\textbf{a}$
(or $\infty$ if no such cycle exists). Observe that a pointed instance
is c-acyclic if and only if its c-girth is $\infty$.

\begin{lemma}[Sparse Incomparability Lemma with Designated Elements] 
\label{lem:sil}
For every pointed instance $(I,\textbf{a})$ and $m>0$,
there is a pointed instance $(I',\textbf{a})$ 
of c-girth at least $m$,
such that $(I',\textbf{a})\to (I,\textbf{a})$ and such 
that, for all pointed instances $(J,\textbf{b})$ of size
at most $m$, $(I,\textbf{a})\to (J,\textbf{b})$
iff $(I',\textbf{a})\to (J,\textbf{b})$.
\end{lemma}

\begin{proof}
Let $(I,\textbf{a})$ be given, with $\textbf{a}=a_1\ldots a_k$,
and where $I$ is an instance over schema $\bfS$. Let
$\widehat{I}$ be the instance over schema $\widehat{\bfS}=\bfS\cup\{Q_1, \ldots, Q_k\}$
that extends $I$ with the unary facts $Q_i(a_i)$.
By the standard version of the sparse incomparability lemma,
there is an $\widehat{\bfS}$-instance $I''$ of girth at least $m$
such that $I''\to \widehat{I}$ and such that, for all $\widehat{\bfS}$-instances
$J$ of size at most $m$, $I''\to J$ iff $\widehat{I}\to J$.
Now, let $I'$ be the $\bfS$-instance obtained from $I''$ by (i) replacing every
element satisfying a unary predicate $Q_i$ by $a_i$, and (ii) dropping
the unary predicates $Q_i$. This operation may introduce new cycles
but it is not hard to see that any such newly introduced short cycle must
pass through one of the designated elements. Therefore, 
$(I',\textbf{a})$ has c-girth at least $m$. Furthermore, 
for all $\bfS$-instances $(J,\textbf{b})$ of size at most $m$,
we have that $(I,\textbf{a})\to (J,\textbf{b})$ iff
$\widehat{I}\to \widehat{J}$ iff $I''\to \widehat{J}$ iff
$(I',\textbf{a})\to (J,\textbf{b})$.
\end{proof}

\thmmonadicduality*

\begin{proof}
The 2 to 1 direction follows immediately from the 
previous theorem. 
From 1 to 2: suppose $F$ has a finite duality w.r.t.~$\Sigma$, 
consisting of $(D_1,\textbf{d}_1), \ldots, (D_n,\textbf{d}_n)$ where $D_i\models \Sigma$,
and let $(A,\textbf{a})\in F$.
Then $(A,\textbf{a})\not\to (D_i,\textbf{d}_i)$. Let $m=s\cdot t$, where $s$ is the number of facts in $A$, 
and $t$ is the maximum number of 
conjuncts in the body of a TGD in $\Sigma$.
By Lemma~\ref{lem:sil}, there is a pointed instance
$(A',\textbf{a})$ of c-girth at least $m$, such that  $(A',\textbf{a})\to (A,\textbf{s})$ and $(A',\textbf{a})\not\to (D_i,\textbf{d}_i)$ for all $i\leq n$.
Since $(A',\textbf{a})\to (A,\textbf{a})$ and $A\models\Sigma$, we have $(P(A'),\textbf{a})\to (P(A),\textbf{a})\to (A,\textbf{a})$ (by Lemma~\ref{lem:edatalog-monotonicity} and Lemma~\ref{lem:capturing}(2)).

We may assume without loss of generality that 
$A'$ contains all unary facts belonging to $P(A')$. 
This is because (i) adding these unary facts does
not change the c-girth of the instance, and
(ii) when extending $A'$ with facts from $P(A')$,
the condition that $(A',\textbf{a})\to (A,\textbf{s})$
is preserved, because $P(A',\textbf{a})\to (A,\textbf{s})$,
(iii) the condition that 
$(A',\textbf{a})\not\to (D_i,\textbf{d}_i)$ is clearly
also preserved when extending $A'$ with additional facts.

We already observed that $(P(A'),\textbf{a})\to  (A,\textbf{a})$. Furthermore, $(P_\Sigma(A'),\textbf{a})\not\to (D_i,\textbf{d}_i)$ (because, otherwise, since $(A',\textbf{a})\subseteq (P_\Sigma(A'),\textbf{a})$,
we would have $(A',\textbf{a})\to (D_i,\textbf{d}_i)$). Since $P_\Sigma(A')\models\Sigma$ (by Lemma~\ref{lem:capturing}(1)) and $(P_\Sigma(A'),\textbf{a})\not\to (D_i,\textbf{d}_i)$, by the duality assumption, some pointed instance in
$F$ maps homomorphically to $(P_\Sigma(A'),\textbf{a})$. In fact, the pointed 
instance in question must be $(A,\textbf{a})$ (otherwise we would obtain
a  contradiction with the fact that the members of $F$ are pairwise 
homomorphically incomparable).  Let
$h:(A,\textbf{a})\to (P_\Sigma(A'),\textbf{a})$. 

Let $(B,\textbf{a})$ be the sub-instance of $(P_\Sigma(A'),\textbf{a})$ that is the 
image of $(A,\textbf{a})$ under $h$.
Since all unary facts in $P_\Sigma(A')$ already
belong to $A'$, and $\Sigma$ is monadic, every
fact in $P_\Sigma(A')$ either belongs to $A'$ or else
can be derived from facts in $A'$ by a single rule application. It follows that there is a sub-instance
$B'$ of $A$ of size at most $|B|\cdot t$, such 
that $B\subseteq P_\Sigma(B')$. Since $|B|\leq s$, 
it follows that $B'$ is c-acyclic. 
Furthermore, 
$(A,\textbf{a})\to (P_\Sigma(B'),\textbf{a})$, and $(P_\Sigma(B'),\textbf{a})\subseteq (P_\Sigma(A'),\textbf{a})\to (A,\textbf{a})$, hence also $(P_\Sigma(B'),\textbf{a})\to (A,\textbf{a})$. Therefore, $(A,\textbf{a})$ is homomorphically equivalent to $(P_\Sigma(B'),\textbf{a})$.
\end{proof}

\thmaboxduality*

\begin{proof}
Let $D$ be a finite set of pointed instances such that
$(F,D)$ is a homomorphism duality, as given by Theorem~\ref{thm:cacyclic-duals}. Let $D'=\{(B,\textbf{b}')\mid (B,\textbf{b})\in D, (B',\iota)\in \Omega_{P_\Sigma}(B), \iota(\textbf{b}')=\textbf{b}\}$.
Note that $D'$ consists of instances that do not necessarily satisfy $\Sigma$. By the same arguments as in the proof of Theorem~\ref{thm:c-acyclic-dualities-tgd} we can show that $(F,D')$ is a 
homomorphism duality in the category $ABox_\Sigma[\bfS]$.
\end{proof}

\section{Proofs for Section~\ref{sec:applications}}

\begin{lemma}
\label{lem:char-dualities}
Let $\bfS$ be any schema and $\Sigma$ any FO theory over $\bfS$.
Let $q$ be any UCQ over $\bfS$, and let $E^+,E^-$ be finite
sets of pointed instances $(I,\textbf{a})$ with $I\models\Sigma$. Then the following are equivalent:
\begin{enumerate}
\item The collection of labeled examples 
$(E^+,E^-)$ uniquely characterizes $q$ w.r.t.~$\Sigma$
\item $q$ fits $(E^+,E^-)$ and  $(E^+,E^-)$ is a finite homomorphism duality w.r.t.~$\Sigma$.
\end{enumerate}
\end{lemma}

\begin{proof}
From 1 to 2, if $(E^+,E^-)$ uniquely characterizes $q$ w.r.t.~$\Sigma$, then, by
definition, $q$ fits $(E^+,E^-)$. Furthermore, it follows that
no pointed instance in $E^+$ admits a homomorphism to a pointed 
instance in $E^-$ (otherwise, it would follow by monotonicity of
UCQs that $q$ does not fit the negative examples).
Next, assume for the sake of
a contradiction that $(E^+,E^-)$ is not a homomorphism duality with respect to $K$.
Then there is a pointed instances $(I,\textbf{a})$ with $I\models\Sigma$ that 
neither belongs to $E^+\upclosure$, not to $E^-\downclosure$.
Let $q_1$ be the union of the canonical CQs of $E^+$ and
let $q_2$ be the union of the canonical CQs of $E^+\cup\{(I,\textbf{a})\}$.
Then $q_1$ and $q_2$ are not equivalent and both fit $(E^+,E^-)$, 
a contradiction.

From 2 to 1, let $q'$ be any UCQ that fits $(E^+,E^-)$. We must show
that $q'$ is equivalent to $q$ w.r.t.~$\Sigma$. Consider
any pointed instance $(I,\textbf{a})$ with $I\models\Sigma$. If
$\textbf{a}\in q(I)$ then $(I,\textbf{a})\in E^+\upclosure$,
therefore $\textbf{a}\in q'(I)$. If, on the other hand,
$\textbf{a}\not\in q(I)$, then $(I,\textbf{a})\not\in E^+\upclosure$, hence $(I,\textbf{a})\in E^-\downclosure$,
hence $\textbf{a}\not\in q'(I)$.
\end{proof}

\thmuniqchar*

\begin{proof}
Let $E^+$ be the set of all pointed instances $(P_\Sigma(I),\textbf{a})$, for $(I,\textbf{a})$ a (c-acyclic) canonical instances of CQs in $q$.
By Theorem~\ref{thm:c-acyclic-dualities-tgd} there is a finite set $E^-$ such that 
$(E^+,E^-)$ is a homomorphism duality w.r.t.~$\Sigma$. It follows by Lemma~\ref{lem:char-dualities} that 
$(E^+,E^-)$ uniquely characterizes $q$ w.r.t.~$\Sigma$.
\end{proof}

\thmuniqcharabox*

\begin{proof}
Let $E^+$ be the set of all (c-acyclic) canonical instances of CQs in $q$. By Theorem~\ref{thm:abox-duality}, there is a finite set $E^-$ such that 
$(E^+,E^-)$ is a finite duality in the category $ABox_\Sigma[\bfS]$.  We claim that 
$(E^+,E^-)$, viewed as a collection of labeled ABox-examples, uniquely characterizes $q$ w.r.t.~$\Sigma$.
The proof is similar as the one for Lemma~\ref{lem:char-dualities}:

It is clear from the construction that $q$ fits $E^+$. 
Take any $(I,\textbf{a})\in E^-$. Since
$(E^+,E^-)$ is a finite duality, 
there is no $(J,\textbf{b})\in E^+$ such that 
$(J,\textbf{b})\to_\Sigma (I,\textbf{a})$.
Equivalently, there is no 
$(J,\textbf{b})\in E^+$ such that
$(P_\Sigma(J),\textbf{b})\to (P_\Sigma(I),\textbf{a})$,
and hence there is no 
$(J,\textbf{b})\in E^+$ such that
$(J,\textbf{b})\to (P_\Sigma(I),\textbf{a})$
Since $E^+$ consists of the canonical CQs of $q$,
this means that $\textbf{a}\not\in q(P_\Sigma(I))$. 
In other words, $q$ fits the negative examples $E^-$. 
Finally, let $q'$ be any UCQ that fits $(E^+,E^-)$,
and let $(I,\textbf{a})$ be any instance satisfying $\Sigma$. 
If $\textbf{a}\in q(I)$, then it follows that 
$(J,\textbf{b})\to_\Sigma (I,\textbf{a})$ 
for $(J,\textbf{b})\in E^+$ the canonical query of the CQ
in question, from which it follows (since $q'$ fits $E^+$) that $\textbf{a}\in q'(P(I))$ and hence (since
$P(I)\leftrightarrow_{adom(I)} I$), $\textbf{a}\in q'(I)$.
If, on the other hand, 
$\textbf{a}\not\in q(I)$, then it follows from the duality
that $(I,\textbf{a})\to_\Sigma (J,\textbf{b})$ for some
$(J,\textbf{b})\in E^-$. From this, it then follows
(since $q'$ fits $E^-$) that $\textbf{a}\not\in q'(P(I))$,
and hence (since
$P(I)\leftrightarrow_{adom(I)} I$), 
$\textbf{a}\not\in q'(I)$.
\end{proof}

\section{Expressive completeness of TAM Datalog (Proof of Theorem~\ref{thm:mso-to-tam})}
\label{app:msototam}

\newcommand{\vertex}{\bullet}
\newcommand{\edge}{\blacktriangledown}

Fix a schema $\bfS$ and let $\bfX=\{X_1,\dots,X_n\}$ (which we can consider to be schema consisting of unary predicates). Following \cite{DalmauKO} we consider the set of formal ``tree-terms'' defined inductively from the following operators. 

\begin{itemize}
\item for every $S\subseteq \bfX$, $\vertex_S$ is a tree-term.
\item for every $R\in\bfS$,  for all tree-terms $t_1,\dots,t_k$ with $k=arity(R)$, and for each $i\in [k]$, $\edge_i^R(t_1,\dots,t_k)$ is a tree-term.
\end{itemize}

We define for each tree-term $t$
an associated pointed tree $(T(t),r(t))$ inductively as follows.

\begin{itemize}
\item If $t=\vertex_S$ then $T(t)$ is the tree containing only one node $v$ (hence $r(t)=v$) and facts $X_i(v)$ for every $X_i\in S$.
\item If $t=\edge_i^R(t_1,\dots,t_k)$ then $T(t)$ is the tree obtained by taking the disjoint union of $T(t_1),\dots,T(t_k)$ and adding fact $f=R(r(t_1),\dots,r(t_k))$. Furthermore, $r(t) = r(t_i)$.
\end{itemize}

\begin{lemma}\label{lem:Vterm-trees}
For every finite connected acyclic pointed $(\bfS\cup\bfX)$-instance $(I,a)$, there is 
a tree-term $t$ such that $(T(t),r(t))$ is
isomorphic to $(I,a)$.
\end{lemma}

\begin{proof}
The proof is by induction on the size of instance $I$, 
as counted by the number of $\bfS$-facts.
The base case of the induction is where $I$ does not
contain any $\bfS$-facts. In this case, it follows
from connectedness that $I$ must be a
single-element structure containing only some
$\bfX$-facts. In this case, the statement clearly holds:
it suffices to take $t$ to be the term $\vertex_S$ where
$S$ is the set of all $X_i\in \bfX$ appearing in $I$. 

If $I$ contains $n$ $\bfS$-facts, with $n>0$, then, by connectedness,
$a$ must appear in at least one $\bfS$-fact, that is,
$I$ contains a fact of the form
$R(a_1, \ldots, a_n)$ where, say, $a_i=a$. 
Let $I'$ be the sub-instance of $I$ where the fact $R(a_1, \ldots, a_n)$ is removed. For each 
$j\leq n$, let $I_j$ be the connected component of $I'$
containing $a_j$. By induction, there is a term 
$t_j$ such that $(T(t_i),r(t_i))$ is isomorphic to $(I_j,a_j)$.
Let $t=\edge^R_i(t_1, \ldots, t_n)$.
Then it is easy to see that $(T(t),r(t))$ is isomorphic
to $(I,a)$.
\end{proof}

An automaton, for present purposes, is a tuple
$(\bfS,\bfX,Q,F,\delta)$
consisting of:
\begin{itemize}
\item schemas $\bfS$, $\bfX$.
\item A finite set $Q$ of states, with a distinguished subset $F\subseteq Q$
\item For every operator $o$ of the form $\vertex_S$ or $\edge_i^R$, of arity, say, $r$ (where we view $\vertex_S$
as a zero-ary operation),
a transition relation $\delta_o\subseteq Q^r\times Q$ 
\end{itemize}

Acceptation is as one would expect. A tree-term $t$ is accepted if we can associate a state $q_{t'}$ to each one of its subterms $t'$ such that $q_t\in F$ and 
the mapping $t'\mapsto q_{t'}$ respects the transition relation (meaning, that if $t'=o(t'_1,\dots,t'_r)$ then $(q_{t'_1},\dots,q_{t'_r},q_{t'})\in\delta_o$.

Given a MSO formula $\phi(x_1,\dots,x_n)$ with schema $\bfS$ we shall consider the following associated formula $\phi'$ defined to be the 
MSO-sentence with schema $\bfS\cup\bfX$ defined as
\begin{equation}\label{eq:addx}\exists x_1,\dots,x_n (\phi(x_1,\dots,x_n)\land\bigwedge_{i=1\ldots n} X_i(x_i))\end{equation}

\begin{lemma}
If $\phi$ is monotone then $\phi'$ is monotone as well.
\end{lemma}
\begin{proof}
Assume that $h:A\rightarrow B$, where $A$ and $B$ are $\bfS\cup\bfX$-instances and assume that $A$ satisfies $\phi'$. Let $x_i\mapsto a_i$ be the instantiation witnessing it. Since the predicates of $\bfX$ do not appear in $\phi$ it follows that the $\bfS$-reduct of $A$ satisfies $\phi(a_1,\dots,a_n)$. Since $\phi$ is monotone it follows that the $\bfS$-reduct of $B$ satisfies $\phi(h(a_1),\dots,h(a_n))$. 
Since $h(a_i)\in B^{X_i}$ for each $1\leq i\leq n$, it follows that $B$ satisfies $\phi'$. 
\end{proof}

\begin{theorem}
\label{th:automaton}
 Let $\phi'$ be a MSO-sentence with schema $\bfS\cup\bfX$. Then there is a finite automaton that accepts the set of all tree-terms $t$
such that $T(t)$ satisfies $\phi'$.
\end{theorem}
\begin{proof}
The proof is entirely standard. For the sake of completeness, we spell out the construction, but
we will omit the correctness argument.
As is customary in the literature on automata theory and MSO, 
we will simplify things by assuming a syntactic normal form 
for MSO-formulas, in which all quantification is second-order.
More precisely, we consider formulas built up
from atomic formulas of the form 
\begin{itemize}
\item $R(X_1, \ldots, X_n)$, treated as a shorthand 
for $\exists x_1, \ldots, x_n(R(x_1, \ldots, x_n)
\land X_1(x_1)\land \cdots\land X_n(x_n))$, 
\item $X_1 \subseteq X_2$, treated as shorthand for 
   $\forall y(X_1(y)\to X_2(y))$, and
\item $Singleton(X)$, treated as shorthand for $\exists x(X(x)\land \forall y(X(y)\to y=x))$
\end{itemize}
using disjunction, negation, and existential second-order quantification. It is easy to construction an automaton
for each of the above atomic formulas. The connectives
are handled by the following standard closure operations
on automata:

\begin{itemize}
\item The union of two automata $(\bfS,\bfX,Q^i,F^i,\delta^i)$ $i=1,2$ (assume that $Q^1$ and $Q^2$ are disjoint) is the automaton $(\bfS,\bfX,Q^1\cup Q^2,F^1\cup F^2,\delta)$ where $\delta_o=\delta^1_o\cup\delta^2_o$.
\item The complement of an automaton $(\bfS,\bfX,Q,F,\delta)$ is 
the (deterministic) automaton
$(\bfS,\bfX,2^Q,F',\delta')$, where
$F'=\{Q'\subseteq Q\mid F\cap Q'=\emptyset\}$ and where
$(Q_1, \ldots, Q_r, Q_{r+1})\in\delta'_o$ iff
$Q_{r+1}=\{q\in Q\mid (q_1, \ldots, q_r,q)\in \delta_o \text{ for some } q_1\in Q_1, \ldots, q_r\in Q_r\}$.
\item The projection of $(\bfS,\bfX,Q,F,\delta)$ to $\bfX'\subseteq\bfX$ is defined to be $(\bfS,\bfX',Q,F,\delta')$, where $\delta'$ is obtained by modifying
$\delta$ in the following way. For every $S'\subseteq \bfX'$, $\delta'_{\vertex_{S'}}=\bigcup_{S\cap \bfX'=S'} \delta_{\vertex_S}$.
\end{itemize}
\end{proof}

\begin{theorem}
\label{th:datalog}
Let $A=(\bfS,\bfX,Q,F,\delta)$ be an automaton. There is a connected Boolean monadic tree-shaped Datalog program $P$ with $\bfS^P_{in}=\bfS\cup\bfX$ such that for all $(\bfS\cup\bfX)$-instances $I$,
the following are equivalent:
\begin{enumerate}
\item $Ans()\in P(I)$.
\item There is some tree-term $t$ accepted by $A$ such that $T(t)\rightarrow I$.
\end{enumerate}

\end{theorem}
\begin{proof}
For every state $q$, $\bfS^P_{aux}$ has a unary symbol $E_q$. Let us describe the rules in $P$:
\begin{itemize}
\item For every $o=\vertex_S$ and every $q\in\delta_o$, $\Sigma^P$ contains the rule with head $E_{q}(x)$ and whose body contains $X_i(x)$ for every $X_i\in S$.
\item For every $o=\edge_i^R$ and every $(q_1,\dots,q_k,q)\in\delta_o$, $\Sigma^P$ contains the rule
\[E_{q}(x_i) \colondash R(x_1,\dots,x_k), E_{q_1}(x_1), \ldots, E_{q_k}(x_k)\]
\item For every $q\in F$, we introduce the rule
\[ Ans() \colondash E_q(x) \]
\end{itemize}

Let $I$ be any $(\bfS\cup\bfX)$-instance. The correctness of the construction follows from the following claim:
\begin{claim*}
The following are equivalent for each $a\in adom(I)$ and $q\in Q$:
\begin{enumerate}
    \item $E_q(a)\in P(I)$
    \item There exists some tree-term $t$ such that (i) $(T(t),r(t))\rightarrow (I,a)$ and (ii) there is a run of $A$ on input $t$ that finishes at state $q$
\end{enumerate} 
\end{claim*}
We omit the proof as it is fairly standard. The $(1)\rightarrow(2)$ direction is proved by induction on the derivation length and the $(2)\rightarrow (1)$ direction is by structural induction on $t$. 
\end{proof}

\thmmsototam*

\begin{proof}
From 1 to 2 is immediate.
From 2 to 3 follows immediately from Lemma~\ref{lem:tree-unfoldings} and 
Lemma~\ref{lem:unfoldings}.
In the remainder, we prove $(3)\rightarrow (1)$. 

Assume that $\phi(x_1, \ldots, x_n)$ is is an MSO formula over $\bfS_{in}$ that satisfies (3). Let $R$ be a fresh
binary relation symbol not in $\bfS_{in}$. 
In particular, $R$ that
does not occur in $\phi$. Let $\bfS = \bfS_{in}\cup\{R\}$.
For the purpose of the next steps of the proof, we will
view $\phi$ as a formula over $\bfS$.
Let $\phi'$ be the  MSO sentence over $\bfS\cup\bfX$ as defined as in \ref{eq:addx} (where $\bfX=\{X_1, \ldots, X_n\}$). Let $A$ be the automaton corresponding to $\phi'$ as in Theorem \ref{th:automaton}, let $P$ as the Boolean connected tree-shaped monadic Datalog program as in Theorem \ref{th:datalog} and let $P'$ by the almost-monadic program corresponding to $P$ as in Theorem~\ref{thm:monadic-reduction}. Inspection of the proof of Theorem~\ref{thm:monadic-reduction} shows that tree-shapedness is preserved, and hence $P'$ is a connected
TAM Datalog program. 

\begin{myclaim}\label{claim1}
$P'$ is equivalent to $\phi$ over connected $\bfS$-instances.
\end{myclaim}

\begin{claimproof}
Assume that $P'$ on a connected $\bfS$-instance $I$ produces $Ans(a_1,\dots,a_n)$. Let $\widehat{I}$ be the connected 
$(\bfS\cup\bfX)$-instance extending $I$  with 
$Q_1(a_1), \dots, Q_n(a_n)$.
Then, it follows that $P(\widehat{I})=true$. Then,
there is some tree-term $t$ accepted by $A$ such that $T(t)\rightarrow \widehat{I}$. It follows that $T(t)$ satisfies $\phi'$. Consequently, 
we have that $\widehat{I}$ satisfies $\phi'$. It follows that $I$ satisfies $\phi(a_1,\dots,a_n)$. Note that for this direction we do not use the full condition of tree-determinacy, only monotonicity.

Conversely, assume that $I$ satisfies $\phi(a_1,\dots,a_n)$. Then by (3) 
\[
    (J,b_1,\dots,b_n)\rightarrow (I,a_1,\dots,a_n)
\]
for some $J$ and $b_1,\dots,b_n$ such that $J$ satisfies $\phi(b_1,\dots,b_n)$. 
Let $\widehat{J}$ be the 
$(\bfS\cup\bfX)$-instance extending $J$  with 
$Q_1(b_1), \dots, Q_n(b_n)$.
Let $t$ be a tree-term such that $T(t)$ is isomorphic to $\widehat{J}$, as given
by Lemma~\ref{lem:Vterm-trees}.
It follows that $A$ accepts $t$. Consequently $P(\widehat{I})=true$. 
It follows that $Ans(a_1,\dots,a_n)$ belongs to $P'(I)$.
\end{claimproof}

Finally, let $P''$ be the TAM Datalog program obtained
from $P'$ by dropping all occurrences of the relation
$R$ (which does not occur in $\phi$) from the body 
of every rule of $P'$. The operation of dropping all
occurrences of $R$ might make some rules unsafe. That is, one or more variable $x$ occurring in the head of a rule might not occur
in the body anymore. This can, however, be easily fixed by
extending the rule body with an $\bfS_{in}$-atom 
containing $x$ and with fresh variables in all other positions of the atom (there are multiple ways to do this, and we add all safe rules that can be 
obtained in this way).
Then it follows from Claim~\ref{claim1} 
that $P''$ is equivalent to $\phi$: 
take any $\bfS_{in}$-instance $I$ and let $I'$ be 
$\bfS$-instance extending $I$ with
all possible $R$-facts over $adom(I)$. Then 
$I\models\phi(a_1, \ldots, a_n)$ iff
$I'\models\phi(a_1, \ldots, a_n)$ iff 
$Ans(a_1, \ldots, a_n)\in P'(I')$ iff
$Ans(a_1, \ldots, a_n)\in P''(I)$.
\end{proof}

\section{Pultr functors as a special case of \texorpdfstring{\edatalog}{existsDatalog}}
\label{app:pultr}

In this appendix, we show that Pultr functors can be
cast as a special case of \edatalog programs. For 
ease of exposition, we follow~\cite{Foniok2015:functors}
in considering only digraph functors. However,
the same argument below extends also to arbitrary
relational structures, as in~\cite{DalmauKO}.

Recall
that a $k$-ary Pultr functor (for digraphs) is specified by a pair $F=(\phi_V,\phi_E)$ where $\phi_V$ is a conjunctive query of arity $k$ and $\phi_E$ is a conjunctive query of arity $2k$.
Both CQs are assumed to be over a signature $\bfS=\{V,E\}$,
where $V$ is unary and $E$ is binary. 
Given a digraph $G=(V,E)$, $F(G)$ is the digraph
whose vertices are all $k$-tuples $\textbf{a}\in adom(V)^k$
satisfying $\phi_V$, and whose edges are all pairs $(\textbf{a},\textbf{b})$ satisfying $\phi_E$. 

\begin{example}
The arc-graph functor is defined to be $F=(\phi_V,\phi_E)$, where $\phi_V(x,y) = E(x,y)$ and $\phi_E(x,y,y,z)=E(x,y)\wedge E(y,z)$. Then for every graph $G$, $F(G)$ is the graph whose node-set are the edges of $G$ and that has an edge joining $(u,v)$ and $(u',v')$ whenever $v=u'$.  The existence of right-adjoint of arc-graph functor has been used in \cite{KrokhinOWZ20} to improve the state-of-the-art in approximate graph coloring.
\end{example}

It follows from the next proposition that every Pultr functor can be simulated by a \edatalog program.

\begin{proposition}
For every Pultr functor $F$ $(\phi_V,\phi_E)$
there is a weakly acyclic $\edatalog$ program $P$
such that for all digraphs $I$, $P(I) \leftrightarrow_{adom(I)} F(I)$.
\end{proposition}

\begin{proof}
  Let $F=(\phi_V,\phi_E)$ where $\phi_V$ be $k$-ary and $\phi_E$ $2k$-ary.
  Take $P=(\bfS_{in},\bfS_{out}, \bfS_{aux}, \Sigma)$ where $\bfS_{in}=\{V_{in},E_{in}\}$,
  $\bfS_{out}=\{V_{out},E_{out}\}$, 
  $\bfS_{aux} = \{R_1, \ldots, R_k\}$ and
  $\Sigma$ consists of:
  \[\begin{array}{lll}
    \exists y \bigwedge_{i=1\ldots k} R_i(y,x_i) &
       \colondash& \phi_V(x_1, \ldots, x_n) \\
    E(u,v) &\colondash& 
       \phi_E(x_1, \ldots, x_k, y_1, \ldots, y_k),
       \bigwedge_{i=1\ldots k} (R_i(u,x_i)\land
        R_i(v,y_i))
    \end{array}\]
 
    It is easy to see that, for all digraphs $I$, $P(I)\leftrightarrow F(I)$. In fact,
    since
    neither $P(I)$ nor $F(I)$ contains any elements
    from $adom(I)$, we have that
    $P(I) \leftrightarrow_{adom(I)} F(I)$.
\end{proof}

\end{document}